\documentclass[11pt,a4paper]{article}

\usepackage{amsfonts}
\usepackage{amsmath}
\usepackage{amssymb}
\usepackage{amsthm}
\usepackage{fullpage}
\usepackage{mleftright}
\usepackage{enumitem}
\usepackage{color}
\usepackage[normalem]{ulem}

\DeclareFontFamily{U}{mathx}{\hyphenchar\font45}
\DeclareFontShape{U}{mathx}{m}{n}{
	<5> <6> <7> <8> <9> <10>
	<10.95> <12> <14.4> <17.28> <20.74> <24.88>
	mathx10
}{}
\DeclareSymbolFont{mathx}{U}{mathx}{m}{n}
\DeclareFontSubstitution{U}{mathx}{m}{n}
\DeclareMathAccent{\widecheck}{0}{mathx}{"71}

\DeclareMathOperator*{\E}{\mathbb{E}}
\DeclareMathOperator*{\Var}{\mathrm{Var}}

\newcommand*{\norm}[1]{\left\|#1\right\|}
\newcommand*{\supp}{\mbox{supp}}
\newcommand*{\dist}{\mbox{dist}}
\newcommand*{\down}{\!\downarrow\!}

\newcommand{\R}{\mathbb{R}}

\newcommand{\F}{\mathbb{F}}
\newcommand*\cupdot{\mathbin{\mathaccent\cdot\cup}}

\numberwithin{equation}{section}
\newtheorem{theorem}{Theorem}[section]
\newtheorem*{theorem*}{Theorem}
\newtheorem{lemma}[theorem]{Lemma}
\newtheorem{claim}[theorem]{Claim}
\newtheorem{proposition}[theorem]{Proposition}
\newtheorem{corollary}[theorem]{Corollary}

\theoremstyle{definition}
\newtheorem{definition}[theorem]{Definition}
\newtheorem*{remark*}{Remark}

\mleftright
\begin{document}

\title{Unique-Neighbor-Like Expansion and\\Group-Independent Cosystolic Expansion\footnote{Originally appeared in ISAAC 2021~\cite{KM21}.}}
\author{
	Tali Kaufman\thanks{Bar-Ilan University, Israel. Email: \texttt{kaufmant@mit.edu}. Supported by ERC and BSF.} \and
	David Mass\thanks{Bar-Ilan University, Israel. Email: \texttt{dudimass@gmail.com}. Supported by the Adams Fellowship Program of the Israel Academy of Sciences and Humanities.}}
\maketitle

\begin{abstract}
	In recent years, high dimensional expanders have been found to have a variety of applications in theoretical computer science, such as efficient CSPs approximations, improved sampling and list-decoding algorithms, and more. Within that, an important high dimensional expansion notion is \emph{cosystolic expansion}, which has found applications in the construction of efficiently decodable quantum codes and in proving lower bounds for CSPs.
	
	Cosystolic expansion is considered with systems of equations over a group where the variables and equations correspond to faces of the complex. Previous works that studied cosystolic expansion were tailored to the specific group $\F_2$. In particular, Kaufman, Kazhdan and Lubotzky (FOCS 2014), and Evra and Kaufman (STOC 2016) in their breakthrough works, who solved a famous open question of Gromov, have studied a notion which we term ``parity'' expansion for small sets. They showed that small sets of $k$-faces have proportionally many $(k+1)$-faces that contain \emph{an odd number} of $k$-faces from the set. Parity expansion for small sets could be used to imply cosystolic expansion only over $\F_2$.
	
	In this work we introduce a stronger \emph{unique-neighbor-like} expansion for small sets. We show that small sets of $k$-faces have proportionally many $(k+1)$-faces that contain \emph{exactly one} $k$-face from the set. This notion is fundamentally stronger than parity expansion and cannot be implied by previous works.
	
	We then show, utilizing the new unique-neighbor-like expansion notion introduced in this work, that cosystolic expansion can be made \emph{group-independent}, i.e., unique-neighbor-like expansion for small sets implies cosystolic expansion \emph{over any group}.
\end{abstract}

\section{Introduction}
\paragraph*{High dimensional expanders.}
High dimensional expanders are a high dimensional analog of expander graphs. A $d$-dimensional simplicial complex is a hypergraph with hyperedges of size at most $d+1$ which is downwards closed, i.e., if $\sigma$ is a hyperedge and $\tau \subset \sigma$ then $\tau$ is also a hyperedge. A hyperedge of size $k+1$ is called a $k$-face of the complex.

In recent years, high dimensional expanders have found a variety of applications in theoretical computer science, such as efficient CSPs approximations~\cite{AJT19}, improved sampling algorithms~\cite{ALGV19, ALG20, AL20, CLV20, CLV20.1, CGSV21, FGYZ21}, improved list-decoding algorithms~\cite{DHKNT19,AJQST20}, sparse agreement tests~\cite{DK17, DD19, KM20} and more.

An especially important high dimensional expansion notion is \emph{cosystolic expansion}. It has been shown to be a key ingredient in the construction of efficiently decodable quantum LDPC codes with a large distance~\cite{EKZ20}, and recently it has been used in the construction of explicit 3XOR instances that are hard for the Sum-of-Squares hierarchy~\cite{DFHT20}.

\paragraph*{Cosystolic expansion as an expanding system of equations.}
A simplicial complex can be viewed as forming \emph{a system of equations} over some group $G$. Consider a $d$-dimensional simplicial complex and some dimension $k < d$. The variables of the system are the $k$-faces of the complex, and the equations are defined by the $(k+1)$-faces; each $(k+1)$-face $\sigma$ corresponds to the equation $\sum_{i=0}^{k+1}\tau_i = 0$, where $\tau_i$ are the $k$-faces contained in $\sigma$ and the sum is performed over the group (e.g., addition modulo $2$ when the group is $\F_2$).

For any assignment of values to the variables which does not satisfy all the equations, there are two measures of interest. One measure is the fraction of unsatisfied equations (out of all the equations), and the second measure is the fraction of variables (out of all the variables) that their value needs to be changed in order to satisfy all the equations. The second measure is also called the \emph{distance} of the assignment from a satisfying assignment.

A system of equations is said to be \emph{expanding} if for any assignment of values to the variables it holds that either all the equations are satisfied or the fraction of unsatisfied equations is proportional to the distance of the assignment from a satisfying assignment. A $d$-dimensional simplicial complex is said to be a \emph{cosystolic expander} over a group $G$ if for all $k < d$, the system of equations formed by its $k$-faces is expanding.

As a simple example, consider a $1$-dimensional simplicial complex (i.e., a graph) and the field $\F_2$. The variables of the system are the vertices of the graph, and the equations are $v_i + v_j = 0 \; (\mbox{mod }2)$ for each edge $\{v_i,v_j\}$. In this case, if the given graph is an expander graph (i.e., each subset of vertices has proportionally many outgoing edges), then the system of equations is expanding. This is true since each assignment of values over $\F_2$ to the vertices can be identified with a subset of vertices, and the unsatisfied equations are exactly the outgoing edges of this set.

\paragraph*{Parity expansion for small sets.}
Kaufman, Kazhdan and Lubotzky~\cite{KKL14}, and Evra and Kaufman~\cite{EK16} in their breakthrough works proved the existence of cosystolic expanders of every dimension, solving a famous open question of Gromov~\cite{Gro10}. In their works they have studied a notion which we term ``parity'' expansion for small sets: They have shown that certain high dimensional expansion properties imply that small sets of $k$-faces have proportionally many $(k+1)$-faces that contain \emph{an odd number} of $k$-faces from the given set. Then they utilized this property in order to imply cosystolic expansion over the group $\F_2$.

\paragraph{$\delta_1$-expansion for small sets.}
In this work we study a fundamentally stronger ``unique-neighbor-like'' expansion in simplicial complexes, which we call \emph{$\delta_1$-expansion}. Let $X$ be a $d$-dimensional simplicial complex and $A$ a set of $k$-faces in $X$. We define $\delta_1(A)$ to be the set of $(k+1)$-faces which contain \emph{exactly one} $k$-face from $A$. We say that $A$ is \emph{$\delta_1$-expanding} if the fraction of $(k+1)$-faces in $\delta_1(A)$ (out of all the $(k+1)$-faces) is proportional to the fraction of $k$-faces in $A$ (out of all the $k$-faces). We show that certain high dimensional expansion properties imply that small sets are $\delta_1$-expanding.

\paragraph{$\delta_1$-expansion and group-independent cosystolic expansion.}
The strength of our $\delta_1$-expansion can be demonstrated by its relation to cosystolic expansion. As explained above, cosystolic expansion is considered with a system of equations over a group. Hence, when proving cosystolic expansion, one has to take the group into account. For instance, previous works could obtain cosystolic expansion only over $\F_2$, because only over $\F_2$ there is an equivalence between an unsatisfied equation and an equation that contains an odd number of non-zero variables.

The $\delta_1$-expansion property that we study in this work has the interesting property that it can make cosystolic expansion to be \emph{group-independent}, i.e., it implies cosystolic expansion over any group. The key point is that an equation with exactly one non-zero variable must be unsatisfied \emph{regardless of the group}. Thus, even though cosystolic expansion is defined over a group, $\delta_1$-expansion implies it over any group.

We further expect that this stronger expansion notion may have implications to quantum codes and CSPs lower bounds.

\paragraph{On the novelty of our work.}
We would like to provide a general outline of the differences between our work and previous works~\cite{KKL14,EK16}.

One fundamental difference is the object we analyze. The major part of previous works is dedicated to the analysis of the expansion of \emph{arbitrary small sets}. In our work, the main analysis is focused on the expansion of \emph{``structured'' small sets} (given by the coboundary of a small set). We use a similar machinery as in previous works, but we leverage the extra structure of the small sets in order to obtain the stronger $\delta_1$-expansion.

We note that it is not trivial how to utilize this extra structure of small sets in order to obtain $\delta_1$-expansion. One cannot just plug it in the proof of previous works and obtain $\delta_1$-expansion. It requires a completely different proof strategy, which we describe next.

Briefly, the proof strategy of~\cite{KKL14} and~\cite{EK16} is as follows. Given a small set $A$ of $k$-faces, they define a notion of ``fat'' faces, where an $\ell$-face, $\ell \le k$, is considered fat if a large fraction of the $k$-faces that contain it belongs to $A$. It is trivial that in dimension $\ell = k$, the set $A$ sits only on fat $k$-faces, since every $k$-face that belongs to $A$ is fat (because the only $k$-face that contains a $k$-face is itself). It is also trivial that in dimension $\ell = -1$, $A$ sits only on thin $(-1)$-faces, since the only $(-1)$-face is the empty set which is contained in all of the $k$-faces of the complex, and $A$ is a small set of $k$-faces. Therefore, there must exist a dimension $\ell \le k$ for which a transition from mostly fat faces to mostly thin faces occurs, i.e., in dimension $\ell$, $A$ sits mostly on fat $\ell$-faces, and in dimension $\ell-1$, $A$ sits mostly on thin $(\ell-1)$-faces. Their argument is then that the fat $\ell$-faces contribute to the parity expansion of $A$, whereas the thin $(\ell-1)$-faces account for a negligible error term.

The proof strategy in our work is essentially the opposite. A fat face, which contributes to the parity expansion in the works of~\cite{KKL14} and~\cite{EK16}, does not have a large $\delta_1$, and hence it is impossible to obtain $\delta_1$-expansion from the fat faces. Our main idea is to gain the $\delta_1$-expansion \emph{from the thin faces}. We observe that if a set sits mostly on thin faces of \emph{one dimension below} then it has a large $\delta_1$. Thus, it is crucial for us to show that a small set $A$ of $k$-faces actually sits mostly on thin faces of one dimension below. Using the terminology of previous paragraph, we have to show that the transition from mostly fat faces to mostly thin faces happens in dimension $k$ itself.

This is where the ``structure'' comes into play. By considering small sets that are obtained as a coboundary of another set, we know that their own coboundary is $0$. Without getting too much into the details, it allows us to bound the fraction of fat faces of dimension $\ell$ by the fraction of fat faces of dimension $\ell-1$, for every $0 \le \ell \le k-1$. Thus, since there are no fat faces in dimension $-1$ (because $A$ is small), we conclude that there are no fat faces at any dimension! Therefore, $A$ sits mostly on thin $(k-1)$-faces and hence has a large $\delta_1$.

\subsection{Some basic definitions}
\paragraph*{Coboundary and cosystolic expansion.}

For the sake of introduction we formally define coboundary and cosystolic expansions only over the field $\F_2$. The general definitions will be given in section 2.

Recall that a $d$-dimensional simplicial complex $X$ is a downwards closed $(d+1)$-hypergraph. A $k$-face of $X$ is a hyperedge of size $k+1$, and the set of $k$-faces of $X$ is denoted by $X(k)$. An assignment of values from $\F_2$ to the $k$-faces, $k\le d$, is called a $k$-cochain, and the space of all $k$-cochains over $\F_2$ is denoted by $C^k(X;\F_2)$.

Any assignment to the $k$-faces $f \in C^k(X;\F_2)$ induces an assignment to the $(k+1)$-faces by the coboundary operator $\delta$. For any $(k+1)$-face $\sigma = \{v_0,\dotsc,v_{k+1}\}$, $\delta(f)(\sigma)$ is defined by
$$\delta(f)(\sigma) = \sum_{i=0}^{k+1}f(\sigma \setminus \{v_i\})\quad (\mbox{mod } 2).$$

We can view the complex as inducing a system of equations, where the equations are determined by the coboundary operator; i.e., each $(k+1)$-face $\sigma \in X(k+1)$ defines the equation $\delta(f)(\sigma) = 0$. The assignments that satisfy all the equations are called the \emph{$k$-cocycles} and denoted by
$$Z^k(X;\F_2) = \{f \in C^k(X;\F_2) \;|\; \delta(f) = \mathbf{0} \}.$$

One can check that $\delta(\delta(f)) = \mathbf{0}$ always holds; i.e., every assignment that is obtained as a coboundary of one dimension below satisfies all the equations. These assignments, that are the coboundary of an assignment of one dimension below, are called the \emph{$k$-coboundaries} and denoted by
$$B^k(X;\F_2) = \{\delta(f) \;|\; f \in C^{k-1}(X;\F_2) \}.$$
Note that $B^k(X;\F_2) \subseteq Z^k(X;\F_2) \subseteq C^k(X;\F_2)$.

For a $d$-dimensional simplicial complex $X$, let $P_d:X(d) \to \R_{\ge 0}$ be a probability distribution over the $d$-faces of the complex. For simplicity, we will assume in this work that $P_d$ is the uniform distribution. This probability distribution over the $d$-faces induces a probability distribution $P_k$ for every dimension $k < d$ by selecting a $d$-face $\sigma_d$ according to $P_d$ and then selecting a $k$-face $\sigma_k \subset \sigma_d$ uniformly at random.

The weight of any $k$-cochain $f \in C^k(X;\F_2)$ is defined by
$$\norm{f} = \Pr_{\sigma \sim P_k}[f(\sigma) \ne 0],$$
i.e., the (weighted) fraction of non-zero elements in $f$. The distance between two $k$-cochains $f,g \in C^k(X;\F_2)$ is defined as $\dist(f,g) = \norm{f - g}$.

We can now introduce the notions of coboundary and cosystolic expansion. As mentioned, a complex is said to be a cosystolic expander if for any assignment that does not satisfy all the equations it holds that the fraction of unsatisfied equations is proportional to the distance of the assignment from a satisfying assignment. Formally:

\begin{definition}[Cosystolic expansion]
	A $d$-dimensional simplicial complex $X$ is said to be an $(\varepsilon,\mu)$-cosystolic expander over $\F_2$, if for every $k < d$:
	\begin{enumerate}
		\item For any $f \in C^k(X;\F_2)\setminus Z^k(X;\F_2)$ it holds that
		$$\frac{\norm{\delta(f)}}{\dist(f,Z^k(X;\F_2))} \ge \varepsilon,$$
		where $\dist(f,Z^k(X;\F_2)) = \min\{\dist(f,g) \;|\; g \in  Z^k(X;\F_2) \}$.
		\item For any $f \in Z^k(X;\F_2)\setminus B^k(X;\F_2)$ it holds that $\norm{f} \ge \mu$.
	\end{enumerate}
\end{definition}

The second condition in the definition ensures that the complex cannot be split into many small pieces, i.e., any satisfying assignment that is not obtained as a coboundary must be large.

Coboundary expansion has been introduced by Linial and Meshulam~\cite{LM06} and independently by Gromov~\cite{Gro10}. It is a similar but stronger notion than cosystolic expansion. The main difference is that the only satisfying assignments in a coboundary expander are coboundaries (unlike cosystolic expansion, where there could be satisfying assignments which are not coboundaries as long as they are large). Formally:

\begin{definition}[Coboundary expansion]
	A $d$-dimensional simplicial complex $X$ is said to be an $\varepsilon$-coboundary expander over $\F_2$ if for every $k < d$ and $f \in C^k(X;\F_2)\setminus B^k(X;\F_2)$ it holds that
	$$\frac{\norm{\delta(f)}}{\dist(f,B^k(X;\F_2))} \ge \varepsilon,$$
	where $\dist(f,B^k(X;\F_2)) = \min\{\dist(f,g) \;|\; g \in  B^k(X;\F_2) \}$.
\end{definition}

\paragraph*{Local spectral expansion.}
Another notion of high dimensional expansion, called \emph{local spectral expansion} is concerned with the spectral properties of local parts of the complex.

For every face $\sigma \in X$, its local view, also called its \emph{link}, is a $(d-|\sigma|-1)$-dimensional simplicial complex defined by $X_\sigma = \{\tau \setminus \sigma \;|\; \sigma \subseteq \tau \in X\}$. The probability distribution over the top faces of $X_\sigma$ is induced from the probability distribution of $X$, where for any top face $\tau \in X_\sigma(d-|\sigma|-1)$, its probability is the probability to choose $\sigma \cup \tau$ in $X$ conditioned on choosing $\sigma$. Since we assume in this work that the probability distribution over the top faces of $X$ is the uniform distribution, it follows that the probability distribution over the top faces of $X_\sigma$ is the uniform distribution.

We can now introduce the notion of a local spectral expander.
\begin{definition}[Local spectral expansion]
	A $d$-dimensional simplicial complex $X$ is called a \emph{$\lambda$-local spectral expander} if for every $k \le d-2$ and $\sigma \in X(k)$, the underlying graph\footnote{The graph whose vertices are $X_\sigma(0)$ and its edges are $X_\sigma(1)$.} of $X_\sigma$ is a $\lambda$-spectral expander.
\end{definition}

\subsection{Summary of main results}
Our main result is a ``unique-neighbor-like'' expansion for non-local small sets, which we call $\delta_1$-expansion. We start with the definition of the $\delta_1$ of a set.

\begin{definition}[$\delta_1$]
	Let $X$ be a $d$-dimensional simplicial complex. For any set of $k$-faces $A \subseteq X(k)$, $0 \le k \le d-1$, we define $\delta_1(A) \subseteq X(k+1)$ to be the set of $(k+1)$-faces that contain exactly one $k$-face from $A$.
\end{definition}

Towards proving that small sets have a large $\delta_1$ we introduce an intermediate notion of \emph{non-local sets}. Roughly speaking, we say that a set of $k$-faces is non-local if its ``local view'' in almost all of the $(k-1)$-faces resemble the global picture.

In order to define this notion of non-local sets, we first define the \emph{localization} of a set to a link of a face. For any set $A\subseteq X(k)$ and an $\ell$-face $\sigma \in X(\ell)$, $\ell < k$, the localization of $A$ to the link of $\sigma$ is a set of $(k-\ell-1)$-faces in the link of $\sigma$ defined by $$A_\sigma = \{\tau \in X_\sigma(k-\ell-1) \;|\; \sigma \cup \tau \in A \}.$$

We also add a useful definition of a mutual weight of two sets. For $\ell < k$ and two sets $A \subseteq X(k), B \subseteq X(\ell)$ we define their mutual weight by
$$\norm{(A,B)} = \Pr_{\sigma_k \sim P_k, \sigma_\ell \subset \sigma_k}[\sigma_k \in A \wedge \sigma_\ell \in B],$$
where $\sigma_k$ is chosen according to the distribution $P_k$ and $\sigma_\ell$ is an $\ell$-face chosen uniformly from $\sigma_k$ (i.e., $\sigma_\ell$ is chosen according to $P_\ell$ conditioned on $\sigma_k$ being chosen). This notion captures how much the sets are related. For instance, if $\norm{(A,B)} \approx \norm{A}$, it means that $A$ contains mostly faces from $B$.

We can now define non-local sets.

\begin{definition}[Non-local sets]
	Let $X$ be a $d$-dimensional simplicial complex and $0 < \eta,\varepsilon<1$. For any set of $k$-faces $A \subseteq X(k)$, $0 \le k \le d-1$, we define the following set of $(k-1)$-faces:
	$$S_{k-1} = \{\sigma \in X(k-1) \;|\; \norm{A_\sigma} \le \eta \}.$$
	We say that $A$ is $(\eta,\varepsilon)$-non-local if $\norm{(A,S_{k-1})} \ge (1-\varepsilon)\norm{A}$.
\end{definition}

As a simple example of a ``local'' set, consider a set of edges $A$ composed of all the edges touching a single vertex. In this case, $\norm{(A, S_0)} = (1/2)\norm{A}$. It can be easily checked that in this example, all triangles contain either $0$ or $2$ edges. As can be seen from this example, local sets are not necessarily $\delta_1$-expanding.

The first theorem we show is that non-local sets are $\delta_1$-expanding.

\begin{theorem}[Non-local sets are $\delta_1$-expanding - informal]\label{thm-1-formal}
	Let $X$ be a $d$-dimensional local spectral expander. For any $A \subseteq X(k)$, $1 \le k \le d-1$, if $A$ is non-local then $\norm{\delta_1(A)} \ge \Omega(\norm{A})$.
\end{theorem}

We consider now a bounded degree local spectral expander whose links are coboundary expanders, where a complex is said to be $q$-bounded degree if every vertex is contained in at most $q$ top faces. We show that every set of unsatisfied equations can be treated as if it is non-local. Specifically, we consider sets of the form $\supp(\delta(f))$ for a $k$-cochain $f \in C^k(X;G)$, over some group $G$. We show a procedure that is given a $k$-cochain $f$ such that $\norm{\delta(f)}$ is small, and returns a $k$-cochain $f'$ which is close to $f$ such that $\delta(f')$ is non-local.

\begin{theorem}[Correction algorithm - informal]\label{thm:correction-algorithm}
	Let $X$ be a $d$-dimensional bounded degree local spectral expander with coboundary expanding links over a group $G$. For any $f \in C^k(X;G)$, $1 \le k \le d-2$, if $\norm{\delta(f)}$ is sufficiently small, then $f$ is close to a $k$-cochain $f' \in C^k(X;G)$ such that $\delta(f')$ is small and non-local. Furthermore, there is an efficient algorithm that is given $f$ and finds $f'$.
\end{theorem}

We conclude by a similar reduction as in~\cite{KKL14} in order to obtain cosystolic expansion over any group. \cite{KKL14} and~\cite{EK16} could obtain cosystolic expansion only over $\F_2$ because their expansion for small sets only guaranteed that they touch many faces of one dimension above an odd number of times. Since we show here $\delta_1$-expansion of such sets, we obtain cosystolic expansion which does not depend on the group.

\begin{theorem}[Cosystolic expansion over any group - informal]\label{thm:cosystolic-expansion}
	Let $X$ be a $d$-dimensional bounded degree local spectral expander with coboundary expanding links over a group $G$. Then the $(d-1)$-skeleton\footnote{The complex which contains the faces of $X$ up to dimension $d-1$.} of $X$ is a cosystolic expander over $G$.
\end{theorem}

A concrete example of simplicial complexes for which our theorems apply to are the famous Ramanujan complexes~\cite{LSV05.1, LSV05.2}, which are the high dimensional analog of the celebrated LPS Ramanujan graphs~\cite{LPS88}. These complexes are local spectral expanders~\cite{EK16} and their links, called spherical buildings, are coboundary expanders~\cite{LMM16}. We note that~\cite{LMM16} proved that spherical buildings are coboundary expanders only over $\F_2$, but their proof can be easily generalized to any abelian group by considering localizations with orientations of $k$-cochains. As for non-abelian groups, \cite{DM19} proved that spherical buildings are coboundary expanders over non-abelian groups as well. For more on Ramanujan complexes, we refer the reader to~\cite{Lub14}.

\begin{corollary}[Ramanujan complexes are cosystolic expanders over any group]
	Let $X$ be a $d$-dimensional Ramanujan complex. If $X$ is sufficiently thick\footnote{The explanation of the ``thickness'' of a Ramanujan complex is out of scope of this paper. It is only important for us that a Ramanujan complex can be made arbitrarily thick in order to satisfy the required criteria.}, then the $(d-1)$-skeleton of $X$ is a cosystolic expander over any group $G$.
\end{corollary}


\subsection{Organization}
In section 2 we provide some required preliminaries. In section 3 we prove the $\delta_1$-expansion and cosystolic expansion results over abelian groups. In section 4 we provide the definitions for cochains over non-abelian groups and we repeat the same process as in section 3, but this time for non-abelian groups. The general strategy is the same for abelian and non-abelian groups, but the details are different, hence we split them into different sections.

\section{Preliminaries}\label{sec:preliminaries}
\paragraph*{Coboundary and cosystolic expansion over abelian groups.}
Let $X$ be a $d$-dimensional simplicial complex and $G$ an abelian group\footnote{For simplicity we deal here only with abelian groups. We discuss non-abelian groups in section 4.}. We first consider an ordered version of the complex and denote it by $\vec{X}$, where
$$\vec{X} = \{(v_0,\dotsc, v_k) \;|\; k \le d,\; \{v_0,\dotsc,v_k\} \in X \},$$
i.e., $\vec{X}$ contains all possible orderings of every face in $X$.

A $k$-cochain over $G$, $k\le d$, is an antisymmetric function $f:\vec{X}(k) \to G$, where $f$ is said to be antisymmetric if for any permutation $\pi \in Sym(k+1)$, $$f((v_{\pi(0)},v_{\pi(1)},\dotsc,v_{\pi(k)})) = sgn(\pi)f((v_0,v_1,\dotsc,v_k)).$$
The space of all $k$-cochains over $G$ is denoted by $C^k(X;G)$.

Any $k$-cochain is an assignment to the $k$-faces and it induces a $(k+1)$-cochain, i.e., an assignment to the $(k+1)$-faces, by the coboundary operator $\delta$. For any ordered $(k+1)$-face $\sigma = (v_0,\dotsc,v_{k+1})$, $\delta(f)(\sigma)$ is defined by
$$\delta(f)(\sigma) = \sum_{i=0}^{k+1}(-1)^if(v_0,\dotsc,v_{i-1},v_{i+1},\dotsc,v_{k+1}),$$
where the sum is performed over the group. It is not hard to check that for every $k$ and $f \in C^k(X;G)$, $\delta(f)$ is antisymmetric, i.e., a $(k+1)$-cochain.

We can view the complex as inducing a system of equations, where the equations are determined by the coboundary operator; i.e., each $(k+1)$-face $\sigma \in X(k+1)$ defines the equation $\delta(f)(\sigma) = 0$ (note that the ordering of the face does not matter for the satisfaction of the equation). The assignments that satisfy all the equations are called the \emph{$k$-cocycles} and denoted by
$$Z^k(X;G) = \{f \in C^k(X;G) \;|\; \delta(f) = \mathbf{0} \}.$$

One can check that $\delta(\delta(f)) = \mathbf{0}$ always holds; i.e., every assignment that is obtained as a coboundary of one dimension below satisfies all the equations. These assignments, that are the coboundary of an assignment of one dimension below, are called the \emph{$k$-coboundaries} and denoted by
$$B^k(X;G) = \{\delta(f) \;|\; f \in C^{k-1}(X;G) \}.$$
Note that $B^k(X;G) \subseteq Z^k(X;G) \subseteq C^k(X;G)$.

Recall that the weight of a $k$-cochain $f \in C^k(X;G)$ is defined by
$$\norm{f} = \Pr_{\sigma \sim P_k}[f(\sigma) \ne 0],$$
i.e., the (weighted) fraction of non-zero elements in $f$. Since the weight of a cochain is dependent only on its non-zero elements, it is often convenient to consider the set $\supp(f)$, i.e., the set of non-zero elements in $f$, and define equivalently
$$\norm{f} = \norm{\supp(f)} = \Pr_{\sigma \sim P_k}[\sigma \in \supp(f)].$$
For simplicity, we might abuse the notation and write $\sigma \in f$ where we mean that $\sigma \in \supp(f)$.

We repeat the definitions of cosystolic and coboundary expansions, but this time for abelian groups.
\begin{definition}[Cosystolic expansion]
	Let $X$ be a $d$-dimensional simplicial complex and $G$ an abelian group. For positive constants $\varepsilon,\mu > 0$, $X$ is called an $(\varepsilon,\mu)$-cosystolic expander over $G$ if for every $k < d$:
	\begin{enumerate}
		\item For any $f \in C^k(X;G)\setminus Z^k(X;G)$ it holds that
		$$\frac{\norm{\delta(f)}}{\dist(f,Z^k(X;G))} \ge \varepsilon,$$
		where $\dist(f,Z^k(X;G)) = \min\{\dist(f,g) \;|\; g \in  Z^k(X;G) \}$.
		\item For any $f \in Z^k(X;G)\setminus B^k(X;G)$ it holds that $\norm{f} \ge \mu$.
	\end{enumerate}
\end{definition}

Coboundary expansion is a similar but stronger notion than cosystolic expansion. The main difference is that the only satisfying assignments in a coboundary expander are coboundaries (unlike cosystolic expansion, where there could be satisfying assignments which are not coboundaries as long as they are large). Formally:

\begin{definition}[Coboundary expansion]
	Let $X$ be a $d$-dimensional simplicial complex and $G$ an abelian group. For a positive constant $\varepsilon > 0$, $X$ is called an $\varepsilon$-coboundary expander over $G$ if for every $k < d$ and $f \in C^k(X;G)\setminus B^k(X;G)$ it holds that
	$$\frac{\norm{\delta(f)}}{\dist(f,B^k(X;G))} \ge \varepsilon,$$
	where $\dist(f,B^k(X;G)) = \min\{\dist(f,g) \;|\; g \in  B^k(X;G) \}$.
\end{definition}

\paragraph*{Links and localization.}
Recall that the link of a $k$-face $\sigma \in X(k)$ is a $(d-k-1)$-dimensional complex defined by $X_\sigma = \{\tau \setminus \sigma \;|\; \sigma \subseteq \tau \in X\}$, where the probability distribution over faces of $X_\sigma$ is induced from the probability distribution over faces of $X$. Since we assume in this work that $P_d$ is the uniform distribution over the $d$-faces of $X$, it follows that the probability distribution over the top faces of $X_\sigma$ is the uniform distribution. In the rest of the paper, we will omit the explicit probability distribution when it is clear from the context.

Recall also that cochains over abelian groups are defined on ordered faces of the complex. For convenience sake, we fix an arbitrary ordering of the faces so that for any face $\sigma \in X$ there is a unique corresponding ordered face $\vec{\sigma} \in \vec{X}$.

For two disjoint ordered faces $\vec{\sigma} = (v_0, \dotsc, v_k)$ and $\vec{\tau} = (u_0, \dotsc, u_\ell)$ we denote their concatenation by $\vec{\sigma\tau} = (v_0,\dotsc,v_k,u_0,\dotsc,u_\ell)$.  For any $k$-face $\sigma \in X(k)$ and a $(k+\ell+1)$-cochain $f \in C^{k+\ell+1}(X;G)$, the \emph{localization} of $f$ to the link of $\sigma$ is an $\ell$-cochain in the link of $\sigma$, $f_\sigma \in C^\ell(X_\sigma;G)$ defined as follows. For any ordered $\ell$-face $\vec{\tau} \in \vec{X_\sigma}(\ell)$: $f_{\sigma}(\vec{\tau}) = f(\vec{\sigma\tau})$, where $\vec{\sigma\tau}$ is the concatenation of $\vec{\sigma}$ (i.e., the unique corresponding ordered face of $\sigma$) and $\vec{\tau}$.

By the law of total probability, the weight of any $k$-cochain $f \in C^k(X;G)$ can be decomposed as a sum of its weight in the links of the $\ell$-faces of the complex:
\begin{lemma}
	Let $X$ be a $d$-dimensional simplicial complex and $G$ an abelian group. For every $f \in C^k(X;G)$, $k \le d$ and $\ell < k$,
	$$\norm{f} = \sum_{\tau \in X(\ell)}\norm{(f,\tau)},$$
	where $\norm{(f,\tau)}$ is the mutual weight of $\supp(f) \subseteq X(k)$ and $\{\tau\} \subseteq X(\ell)$.
\end{lemma}
\begin{proof}
	It follows immediately from the definitions:
	$$\norm{f} = \Pr_{\sigma\sim P_k}[\sigma \in \supp(f)] =
	\sum_{\tau \in X(\ell)}\Pr_{\sigma\sim P_k, \tau' \subset \sigma}[\sigma \in \supp(f) \wedge \tau' = \tau] =
	\sum_{\tau \in X(\ell)}\norm{(f,\tau)},	$$
	where the second equality follows from the law of total probability.
\end{proof}

\paragraph*{Minimal and locally minimal cochains.}
One of the technical notions we use in this work is the notion of a minimal cochain. We say that a $k$-cochain $f \in C^k(X;G)$ is \emph{minimal} if its weight cannot be reduced by adding a coboundary to it, i.e., for every $g \in B^k(X;G)$ it holds that $\norm{f} \le \norm{f-g}$. Recall that the distance of $f$ from the coboundaries is defined by
$\dist(f, B^k(X;G)) = \min \{\norm{f-g} \;|\; g \in B^k(X;G) \}$. Since $\mathbf{0} \in B^k(X;G)$, it follows that for every $f \in C^k(X;G)$, $\norm{f} \ge \dist(f,B^k(X;G))$. Hence, $f$ is said to be \emph{minimal} if and only if $\norm{f} = \dist(f, B^k(X;G))$.

We also define the notion of a locally minimal cochain, where we say that $f \in C^k(X;G)$ is \emph{locally minimal} if for every vertex $v$, the localization of $f$ to the link of $v$ is minimal in the link, i.e., $f_v$ is minimal in $X_v$ for every $v \in X(0)$.

\paragraph*{Cheeger inequality for graphs.}
A $1$-dimensional simplicial complex $X$ is just a graph. In this case the known Cheeger inequality gives the following (see e.g.~\cite{KKL14} for a proof):
\begin{lemma}\label{many-outgoing-edges}
	Let $X$ be a $1$-dimensional simplicial complex which is a $\lambda$-spectral expander graph. For any set of vertices $A \subseteq X(0)$ it holds that
	\begin{enumerate}
		\item $\displaystyle\norm{E(A,\overline{A})} \ge 2(1-\lambda)\norm{A}\|\overline{A}\|$,
		\item $\displaystyle\norm{E(A)} \le (\norm{A} + \lambda)\norm{A}$,
	\end{enumerate}
	where $E(A,\overline{A})$ is the set of edges with one endpoint in $A$ and one endpoint in $\overline{A}$, and $E(A)$ is the set of edges with both endpoints in $A$.
\end{lemma}

\section{Result for abelian groups}
\subsection{Non-local sets are $\delta_1$-expanding}
In this section we show our results for abelian groups.

Our first theorem is that non-local sets in a local spectral expander have $\delta_1$ that is proportional to their size. We prove theorem~\ref{thm-1-formal} which we restate here in a formal way.
\begin{theorem}\label{balanced-implies-large-delta1-formal}[Non-local sets are $\delta_1$-expanding]
	Let $X$ be a $d$-dimensional $\lambda$-local spectral expander and $0<\eta, \varepsilon<1$. For any $A \subseteq X(k)$, $0 \le k \le d-1$, such that $A$ is $(\eta,\varepsilon)$-non-local it holds that
	$$\norm{\delta_1(A)} \ge \left(1 - \binom{k+2}{k}\Big(\lambda+\eta+2\varepsilon\Big)\right)\norm{A}.$$
\end{theorem}

\begin{proof}
	Recall that we denote by $S_{k-1}$ the set of $(k-1)$-faces $\sigma \in X(k-1)$ satisfying $\norm{A_\sigma} \le \eta$. Let us define the following sets of $(k+1)$-faces:
	\begin{itemize}
		\item $\Gamma(A) = \{\tau \in X(k+1) \;|\; \exists \sigma \in A \mbox{ s.t. } \sigma \subset \tau\}$.
		\item
		$\Gamma(A,\overline{S_{k-1}}) = \left\{\tau \in X(k+1) \;|\;
		\exists \sigma \in A,\;\sigma' \in \overline{{S_{k-1}}} \mbox{ s.t. } \sigma' \subset \sigma \subset \tau\right\}.$
		\item $\Upsilon = \left\{\tau \in X(k+1) \;|\;
		\exists \sigma,\sigma'\in A \mbox{ s.t. }\sigma,\sigma' \subset \tau,\; \sigma \cap \sigma' \in S_{k-1}\right\}.$
	\end{itemize}
	
	In words: $\Gamma(A)$ is the set of all $(k+1)$-faces that contain a $k$-face from $A$, $\Gamma(A,\overline{S_{k-1}})$ is the set of all $(k+1)$-faces that contain a $k$-face from $A$ which contains a $(k-1)$-face from $\overline{S_{k-1}} = X(k-1) \setminus S_{k-1}$, and $\Upsilon$ is the set of all $(k+1)$-faces that contain two $k$-faces from $A$ such that their intersection is a $(k-1)$-face from $S_{k-1}$.
	
	Note that for every $\tau \in \Gamma(A)\setminus\Gamma(A,\overline{S_{k-1}})$ one of the following cases must hold: Either $\tau$ contains exactly one $k$-face from $A$, i.e., $\tau \in \delta_1(A)$, or $\tau$ contains at least two $k$-faces from $A$ such that their intersection belongs to $S_{k-1}$, i.e., $\tau \in \Upsilon$. It follows that
	\begin{equation}\label{balanced-implies-large-delta1-eq1}
	\norm{\delta_1(A)} \ge
	\norm{\Gamma(A)\setminus(\Gamma(A,\overline{S_{k-1}})\cup\Upsilon)} \ge
	\norm{\Gamma(A)} - \norm{\Gamma(A,\overline{S_{k-1}})} - \norm{\Upsilon}.
	\end{equation}
	
	Let us bound each of the above terms separately. First,
	\begin{equation}\label{balanced-implies-large-delta1-eq2}
	\begin{aligned}
	\norm{\Gamma(A)} &=
	\Pr[\sigma_{k+1} \in \Gamma(A)] \\[5pt]&\ge
	\Pr[\sigma_{k+1} \in \Gamma(A) \wedge \sigma_k \in A] \\[5pt]&=
	\Pr[\sigma_k \in A]\cdot\Pr[\sigma_{k+1} \in \Gamma(A) \;|\; \sigma_k \in A] \\[5pt]&=
	\Pr[\sigma_k \in A] = \norm{A}.
	\end{aligned}
	\end{equation}
	
	Second,
	\begin{equation}\label{balanced-implies-large-delta1-eq3}
	\begin{aligned}
	\norm{\Gamma(A,\overline{S_{k-1}})} &=
	\Pr[\sigma_{k+1} \in \Gamma(A,\overline{S_{k-1}})] \\[6pt]&=
	\frac{\Pr[\sigma_{k+1} \in \Gamma(A,\overline{S_{k-1}}) \wedge \sigma_k \in A \wedge \sigma_{k-1} \notin S_{k-1}]}{\Pr[\sigma_k \in A \wedge \sigma_{k-1} \notin S_{k-1} \;|\; \sigma_{k+1} \in \Gamma(f,\overline{S_{k-1}})]} \\[6pt]& \le
	(k+2)(k+1)\Pr[\sigma_k \in A \wedge \sigma_{k-1} \notin S_{k-1}] \\[6pt]&=
	(k+2)(k+1)\norm{(A,\overline{S_{k-1}})} \\[6pt]&\le
	(k+2)(k+1)\varepsilon\norm{A},
	\end{aligned}
	\end{equation}
	where the first inequality holds since the probability that $\sigma_k \in A$ and $\sigma_{k-1} \notin S_{k-1}$ given that $\sigma_{k+1} \in \Gamma(A,\overline{S_{k-1}})$ is at least $1/\big((k+2)(k+1)\big)$, and the second inequality follows since $A$ is an $(\eta,\varepsilon)$-non-local set.
	
	Lastly, consider a $(k+1)$-face $\tau \in \Upsilon$. By definition, $\tau$ contains two $k$-faces $\sigma,\sigma' \in A$ such that $\sigma \cap \sigma' \in S_{k-1}$. Let us denote $\widecheck{\tau} = \sigma \cap \sigma'$. Note that $\tau$ is seen in the link of $\widecheck{\tau}$ as an edge between two vertices in $A_{\widecheck{\tau}}$, i.e., $\tau \setminus \widecheck{\tau} \in E(A_{\widecheck{\tau}})$. Thus,
	\begin{equation}\label{balanced-implies-large-delta1-eq4}
	\begin{aligned}
	\norm{\Upsilon} &=
	\sum_{\tau \in \Upsilon}\Pr[\sigma_{k+1} = \tau] \\[5pt]&=
	\sum_{\tau \in \Upsilon}\frac{\Pr[\sigma_{k+1} = \tau \wedge \sigma_{k-1} = \widecheck{\tau}]}{\Pr[\sigma_{k-1} = \widecheck{\tau} \;|\; \sigma_{k+1} = \tau]} \\[5pt]&\le
	\sum_{\tau \in \Upsilon}\binom{k+2}{k}\Pr[\sigma_{k+1} = \tau \;|\; \sigma_{k-1} = \widecheck{\tau}]\cdot\Pr[\sigma_{k-1} = \widecheck{\tau}] \\[5pt]&\le
	\binom{k+2}{k}\sum_{\sigma \in S_{k-1}}\norm{E(A_\sigma)}\cdot\Pr[\sigma_{k-1} = \sigma] \\[5pt]&\le
	\binom{k+2}{k}\sum_{\sigma \in S_{k-1}}(\norm{A_\sigma} + \lambda)\norm{A_\sigma}\cdot\Pr[\sigma_{k-1} = \sigma] \\[5pt]&\le
	\binom{k+2}{k}(\eta + \lambda)\sum_{\sigma \in S_{k-1}}\norm{(A_\sigma,\sigma)} \le
	\binom{k+2}{k}(\eta + \lambda)\norm{A},
	\end{aligned}
	\end{equation}
	where the third inequality follows since $X$ is a $\lambda$-local spectral expander, and the last inequality follows since $\sigma \in S_{k-1}$.
	
	Substituting~\eqref{balanced-implies-large-delta1-eq2},~\eqref{balanced-implies-large-delta1-eq3} and~\eqref{balanced-implies-large-delta1-eq4} in~\eqref{balanced-implies-large-delta1-eq1} finishes the proof.
\end{proof}

An immediate corollary of theorem~\ref{balanced-implies-large-delta1-formal} is that any non-local cocycle must be zero.
\begin{corollary}\label{balanced-cocycles-vanish}[Non-local cocycles vanish]
	For any $d \in \mathbb{N}$, an abelian group $G$, and $0 < \lambda, \eta, \varepsilon < 1$ such that $\lambda + \eta + 2\varepsilon \le 2/(d+1)^2$ the following holds: Let $X$ be a $d$-dimensional $\lambda$-local spectral expander. For any $f \in Z^k(X;G)$, $0 \le k \le d-1$, if $f$ is $(\eta,\varepsilon)$-non-local then $f = \mathbf{0}$.
\end{corollary}
\begin{proof}
	Since $f$ is $(\eta,\varepsilon)$-non-local, then by theorem~\ref{balanced-implies-large-delta1-formal} it holds that
	\begin{equation*}
	\norm{\delta_1(f)} \ge \frac{\norm{f}}{d+1}.
	\end{equation*}
	
	On the other hand, $f \in Z^k(X;G)$ and hence $\norm{\delta_1(f)} \le \norm{\delta(f)} = 0$. It follows that $f = \mathbf{0}$ as required.
\end{proof}

\subsection{The correction procedure}
Our aim now is to show a correction procedure for small coboundaries. We show an algorithm that gets a cochain $f$ such that $\norm{\delta(f)}$ is small and returns a cochain $f'$ by making a few changes to $f$ such that $\delta(f')$ is non-local.

We start by showing that any small and locally minimal cocycle is non-local.

\begin{proposition}\label{small-locally-minimal-cocycles-are-balanced}[Small and locally minimal cocycles are non-local]
	For any $d \in \mathbb{N}$, an abelian group $G$, and $0<\beta,\varepsilon<1$ there exist $0 < \lambda,\eta \le \varepsilon$ such that the following holds: Let $X$ be a $d$-dimensional $\lambda$-local spectral expander with $\beta$-coboundary expanding links. For any $f \in Z^k(X;G)$, $1 \le k \le d-1$, if $\norm{f} \le \eta^{2^{k+1}-1}$ and locally minimal then $f$ is $(\eta,\varepsilon)$-non-local.
\end{proposition}

In order to prove proposition~\ref{small-locally-minimal-cocycles-are-balanced} we need a few more definitions and lemmas. For the sake of better readability, we just present the required lemmas and postpone their proofs to the end of this subsection. Let $f \in C^k(X;G)$, $0 \le k \le d-1$. Recall that $S_{k-1}$ is the set of $(k-1)$-faces $\sigma$ satisfying $\norm{f_\sigma} \le \eta$. For any $-1 \le i \le k-2$, we define the following set of $i$-faces:
$$S_i = \{\sigma \in X(i) \;|\; (\overline{S_{i+1}})_\sigma \le \eta^{2^{k-i-1}} \}.$$

We will show first that if $\norm{f}$ is sufficiently small then $\norm{S_{-1}} = 1$, i.e., the empty set belongs to $S_{-1}$.

\begin{lemma}\label{empty-set-is-balanced}
	Let $X$ be a $d$-dimensional simplicial complex, $G$ an abelian group and $0 < \eta < 1$. For any $f \in C^k(X;G)$, $0 \le k \le d$, if $\norm{f} \le \eta^{2^{k+1}-1}$ then $\norm{S_{-1}} = 1$.
\end{lemma}

Next, let $\Upsilon \subseteq X(k+1)$ be the set of $(k+1)$-faces which contain two $i$-faces $\sigma, \sigma' \in \overline{S_i}$ such that $\sigma \cap \sigma' \in S_{i-1}$. We then show that $\norm{\Upsilon}$ is a negligible fraction of $\norm{f}$.

\begin{lemma}\label{negligible-degenerate-faces}
	Let $X$ be a $d$-dimensional $\lambda$-local spectral expander, $G$ an abelian group and $0 < \eta < 1$ such that $\lambda \le \eta^{2^{d-1}}$. For any $f \in C^k(X;G)$, $0 \le k \le d-1$, it holds that
	$$\norm{\Upsilon} \le \eta\binom{k+2}{2}2^{k+2}\norm{f}.$$
\end{lemma}

Lastly, for any $\sigma \in X(i)$, we denote by $f\down\sigma$ the set of $k$-faces $\tau \in f$ which have a sequence of containments of faces from $\overline{S_j}$, $i < j < k$, down to $\sigma$, formally,
$$f\down \sigma = \{\tau \in f \;|\; \exists \tau_{k-1}\in \overline{S_{k-1}},\dotsc,\tau_{i+1}\in \overline{S_{i+1}} \mbox{ s.t. }\tau \supset \tau_{k-1} \supset \dotsb \supset \tau_{i+1} \supset \sigma \}.$$

We will show that for any cocycle $f \in Z^k(X;G)$ and $0 \le i \le k-1$, the fraction of $f$ that sits on $i$-faces from $\overline{S_i}$ is approximately the fraction of $f$ that sits on $(i-1)$-faces from $\overline{S_{i-1}}$.

\begin{lemma}\label{small-contribution-of-fat-faces}
	Let $X$ be a $d$-dimensional simplicial complex such that its links are $\beta$-coboundary expanders over an abelian group $G$. For any locally minimal $f \in Z^k(X;G)$, $1 \le k \le d-1$, and $0 \le i \le k-1$ it holds that
	$$\sum_{\sigma \in \overline{S_i}}\norm{(f\!\downarrow\!\sigma, \sigma)} \le
	\frac{1}{\beta}\left((k+1-i)(i+1)\sum_{\sigma' \in \overline{S_{i-1}}}\norm{(f\!\downarrow\!\sigma', \sigma')} + \norm{\Upsilon}\right).$$
\end{lemma}

With the above lemmas in hand, we can now prove proposition~\ref{small-locally-minimal-cocycles-are-balanced}.
\begin{proof}[Proof of proposition~\ref{small-locally-minimal-cocycles-are-balanced}]
	Let
	$$\eta = \frac{\beta^{d-1}\varepsilon}{2^d((d+1)!)^2} \quad\quad\mbox{and}\quad\quad
	\lambda = \eta^{2^{d-1}}.$$
	Applying lemma~\ref{small-contribution-of-fat-faces} on dimensions $i = k-1,k-2,\dotsc,0$ step after step yields
	\begin{equation}\label{small-cocycles-are-balanced-eq1}
	\begin{aligned}
	\norm{(f,\overline{S_{k-1}})} &= 
	\sum_{\sigma \in \overline{S_{k-1}}}\norm{(f\down\sigma,\sigma)} \\&\le
	\frac{1}{\beta}\norm{\Upsilon} + \frac{1}{\beta}\cdot 2\cdot k\sum_{\sigma \in \overline{S_{k-2}}}\norm{(f\down\sigma,\sigma)} \\&\le
	\frac{1}{\beta}\norm{\Upsilon} + \frac{1}{\beta^2}\cdot 2\cdot k\norm{\Upsilon} + \frac{1}{\beta^2}\cdot 2\cdot k\cdot 3\cdot(k-1)\sum_{\sigma \in \overline{S_{k-3}}}\norm{(f\down\sigma,\sigma)} \\&\le
	\left(\frac{1}{\beta} + \frac{1}{\beta^2}\cdot 2\cdot k + \dotsb + \frac{1}{\beta^k}(k!)^2\right)\norm{\Upsilon} + \frac{1}{\beta^k}(k!)^2(k+1)\sum_{\sigma \in \overline{S_{-1}}}\norm{(f\down\sigma,\sigma)} \\&=
	\left(\frac{1}{\beta} + \frac{1}{\beta^2}\cdot 2\cdot k + \dotsb + \frac{1}{\beta^k}(k!)^2\right)\norm{\Upsilon} \le
	k\beta^{-k}(k!)^2\norm{\Upsilon},
	\end{aligned}
	\end{equation}
where the last equality follows by lemma~\ref{empty-set-is-balanced}. Substituting lemma~\ref{negligible-degenerate-faces} in~\eqref{small-cocycles-are-balanced-eq1} completes the proof.
\end{proof}

Now, the idea of the correction algorithm is to make $\delta(f)$ locally minimal by correcting $f$ in a few local parts. The algorithm runs in iterations, where at every iteration it does the following one step of correction.
\begin{lemma}[One step of correction]\label{one-step-of-correction}
	Let $X$ be a $d$-dimensional simplicial complex and $G$ an abelian group. For any $f \in C^k(X;G)$, $1 \le k \le d-1$, if $f$ is not locally minimal then there exists a vertex $v \in X(0)$ and a $(k-1)$-cochain $g \in C^{k-1}(X;G)$ such that $\norm{g} \le k\norm{v}$ and $\norm{f-\delta(g)} < \norm{f}$.
\end{lemma}

\begin{proof}
	Since $f$ is not locally minimal, there exists a vertex $v \in X(0)$ such that $f_v$ is not minimal in $X_v$. By definition there exists a $(k-2)$-cochain $h \in C^{k-2}(X_v;G)$ in the link of $v$ such that $\norm{f_v - \delta(h)} < \norm{f_v}$. Define $g \in C^{k-1}(X;G)$ by
	$$g(\sigma) = \begin{cases}
	h(\tau) & \sigma = v\tau,\\
	0 & \mbox{otherwise}.
	\end{cases}
	$$
	
	Note that $g_v = h$, therefore $\norm{f - \delta(g)} < \norm{f}$.
	Furthermore, since $g(\sigma) = 0$ for every $\sigma$ which does not contain $v$ it follows that
	\begin{equation*}
	\norm{g} = \Pr[\sigma_{k-1} \in g] =
	\frac{\Pr[\sigma_{k-1} \in g \wedge \sigma_0 = v]}{\Pr[\sigma_0 = v \;|\; \sigma_{k-1} \in g]} \le k\norm{v}.
	\end{equation*}
\end{proof}

We then use lemma~\ref{one-step-of-correction} iteratively in order to prove theorem~\ref{thm:correction-algorithm} which we restate here in a formal way.
\begin{theorem}[Correction algorithm]\label{balancing-algorithm}
	For any $d,q \in \mathbb{N}$, an abelian group $G$, and $0<\beta,\varepsilon<1$ there exist constants $0 < \lambda,\eta \le \varepsilon$ such that the following holds: Let $X$ be a $d$-dimensional $q$-bounded degree $\lambda$-local spectral expander with $\beta$-coboundary expanding links. For any $f \in C^k(X;G)$, $1 \le k \le d-2$, if $\norm{\delta(f)} \le \eta^{2^{k+2}-1}$ then there exists $f' \in C^k(X;G)$ such that $\dist(f, f') \le q\binom{d}{k+1}\norm{\delta(f)}$, $\norm{\delta(f')} \le \norm{\delta(f)}$, and $\delta(f')$ is $(\eta,\varepsilon)$-non-local.
\end{theorem}

\begin{proof}
	Let $\lambda$ and $\eta$ be as in proposition~\ref{small-locally-minimal-cocycles-are-balanced}.
	Apply lemma~\ref{one-step-of-correction} for $\delta(f)$ step by step until no more corrections are possible. Since at every step the weight decreases, this process terminates after some $r \ge 0$ steps. Denote by $v_1,v_2,\dotsc,v_r$ the vertices and by $g_1,g_2,\dotsc,g_r$ the $k$-cochains given by applying lemma~\ref{one-step-of-correction} for $r$ steps, where at step $i$ we apply it for $\delta(f-g_1-\dotsb-g_{i-1})$.
	
	Let $f' = f - g_1 - g_2 - \dotsb - g_r$. Since the norm of $\delta(f)$ decreases at every step of correction, it follows that $\norm{\delta(f')} \le \norm{\delta(f)} \le \eta^{2^{k+2}-1}$. Furthermore, since no more corrections are possible, it must be that $\delta(f')$ is locally minimal. Thus, by proposition~\ref{small-locally-minimal-cocycles-are-balanced}, $\delta(f')$ is $(\eta,\varepsilon)$-non-local.
	
	It is left to show that $\norm{f-f'}$ is proportional to $\norm{\delta(f)}$. By definition, for any $\sigma \in X(k+1)$ it holds that $\norm{\sigma} \ge \Big(|X(d)|\binom{d+1}{k+2}\Big)^{-1}$, hence $r \le |X(d)|\binom{d+1}{k+2}\norm{\delta(f)}$. Thus,
	\begin{align*}
	\dist(f,f') &=
	\norm{g_1 + g_2 + \dotsb + g_r} \le
	\sum_{i=1}^r(k+1)\norm{v_i} \\&\le
	|X(d)|\binom{d+1}{k+2}(k+1)\frac{q}{|X(d)|(d+1)}\norm{\delta(f)} \le
	\binom{d}{k+1}q\norm{\delta(f)}
	\end{align*}
	which finishes the proof.
\end{proof}

\subsubsection{Proofs of intermediate lemmas}
We will prove now the lemmas we used for proving proposition~\ref{small-locally-minimal-cocycles-are-balanced}.

The proof of lemma~\ref{empty-set-is-balanced} follows immediately from the following claim.
\begin{claim}\label{bound-non-balanced-faces}
	Let $X$ be a $d$-dimensional simplicial complex, $G$ an abelian group and $0 < \eta < 1$. For any $f \in C^k(X;G)$, $0 \le k \le d$, and $-1 \le i \le k-1$ it holds that $\norm{\overline{S_{i}}} < \eta^{1-2^{k-i}}\norm{f}$.
\end{claim}
\begin{proof}
	By laws of probability
	\begin{align*}
	\norm{\overline{S_i}} &= \Pr[\sigma_i \in \overline{S_i}] =
	\frac{\Pr[\sigma_i \in \overline{S_i} \wedge \sigma_{i+1} \in \overline{S_{i+1}}]}{\Pr[\sigma_{i+1} \in \overline{S_{i+1}} \;|\; \sigma_i \in \overline{S_i}]} \\[6pt]&<
	\frac{\Pr[\sigma_{i+1} \in \overline{S_{i+1}}]}{\eta^{2^{k-i-1}}} =
	\frac{\Pr[\sigma_{i+1} \in \overline{S_{i+1}} \wedge \sigma_{i+2} \in \overline{S_{i+2}}]}{\eta^{2^{k-i-1}}\cdot\Pr[\sigma_{i+2} \in \overline{S_{i+2}} \;|\; \sigma_{i+1} \in \overline{S_{i+1}}]} \\[6pt]&<
	\frac{\Pr[\sigma_{i+2} \in \overline{S_{i+2}}]}{\eta^{2^{k-i-1}}\eta^{2^{k-i-2}}} <\dotsb<
	\frac{\Pr[\sigma_k \in f]}{\eta^{2^{k-i-1}}\eta^{2^{k-i-2}}\dotsb\eta} =
	\frac{\norm{f}}{\eta^{2^{k-i}-1}}.
	\end{align*}
\end{proof}

\begin{proof}[Proof of lemma~\ref{empty-set-is-balanced}]
	By claim~\ref{bound-non-balanced-faces}, $\norm{\overline{S_{-1}}} < \eta^{1-2^{k+1}}\norm{f} \le 1$.
	Therefore, $\norm{S_{-1}} = 1-\norm{\overline{S_{-1}}} > 0$, but $X(-1)$ contains only one face, so $\norm{S_{-1}} \in \{0,1\}$. It follows that $\norm{S_{-1}} = 1$ as required.
\end{proof}

\begin{proof}[Proof of lemma~\ref{negligible-degenerate-faces}]
	By definition, any $(k+1)$-face $\tau \in \Upsilon$ contains at least one pair of $i$-faces $\sigma,\sigma' \in \overline{S_i}$ such that $\sigma \cup \sigma' \in X(i+1)$ and $\sigma \cap \sigma' \in S_{i-1}$ for some $0 \le i \le k$. For any $\tau \in \Upsilon$, fix one such pair $\sigma,\sigma' \in \overline{S_i}$ and denote by $\widehat{\tau} = \sigma \cup \sigma'$ and by $\widecheck{\tau} = \sigma \cap \sigma'$. Denote by $\Upsilon_i = \{\tau \in \Upsilon \;|\; \widehat{\tau} \in X(i) \}$, so $\Upsilon$ can be decomposed to $\Upsilon = \bigsqcup_{i=1}^{k+1}\Upsilon_i$. Now, for any $\tau \in \Upsilon_i$
	
	\begin{equation}\label{negligible-degenerate-faces-eq1}
	\begin{aligned}
	\norm{(\widehat{\tau},\widecheck{\tau})} &=
	\Pr[\sigma_i = \widehat{\tau} \;|\; \sigma_{i-2} = \widecheck{\tau}]\cdot\Pr[\sigma_{i-2} = \widecheck{\tau}] \\[5pt]&\le
	\norm{E((\overline{S_{i-1}})_{\widecheck{\tau}})}\cdot\Pr[\sigma_{i-2} = \widecheck{\tau}] \\[5pt]&\le
	(\eta^{2^{k-i+1}} + \lambda)\norm{(\overline{S_{i-1}})_{\widecheck{\tau}}}\cdot\Pr[\sigma_{i-2} = \widecheck{\tau}] \\[5pt]&\le
	(\eta^{2^{k-i+1}} + \lambda)\norm{(\overline{S_{i-1}}, \widecheck{\tau})},
	\end{aligned}
	\end{equation}
	where the first inequality follows since $\widehat{\tau}$ is seen in the link of $\widecheck{\tau}$ as an edge between two vertices from $\overline{S_{i-1}}$ and the rest follow by definitions. Therefore,
	\begin{equation*}
	\begin{aligned}
	\norm{\Upsilon_i} &=
	\sum_{\tau \in \Upsilon_i}\Pr[\sigma_{k+1} = \tau] \\[5pt]&=
	\sum_{\tau \in \Upsilon_i}\frac{\Pr[\sigma_{k+1} = \tau \wedge \sigma_i = \widehat{\tau} \wedge \sigma_{i-2} = \widecheck{\tau}]}{\Pr[\sigma_i = \widehat{\tau} \wedge \sigma_{i-2} = \widecheck{\tau} \;|\; \sigma_{k+1} = \tau]} \\[5pt]&\le
	\sum_{\tau \in \Upsilon_i}\binom{k+2}{i+1}\binom{i+1}{i-1}\norm{(\widehat{\tau},\widecheck{\tau})} \\[5pt]&\le
	\binom{k+2}{i+1}\binom{i+1}{i-1}\sum_{\sigma \in S_{i-2}}(\eta^{2^{k-i+1}} + \lambda)\norm{(\overline{S_{i-1}}, \sigma)} \\[5pt]&\le
	\binom{k+2}{i+1}\binom{i+1}{i-1}(\eta^{2^{k-i+1}} + \lambda)\norm{\overline{S_{i-1}}} \\[5pt]&\le
	\binom{k+2}{i+1}\binom{i+1}{i-1}(\eta^{2^{k-i+1}} + \lambda)\eta^{1-2^{k-i+1}}\norm{f} \\[5pt]&\le
	\binom{k+2}{i+1}\binom{i+1}{i-1}2\eta\norm{f},
	\end{aligned}
	\end{equation*}
	where the second inequality follows by~\eqref{negligible-degenerate-faces-eq1} and the fact that for every $\tau \in \Upsilon_i$ we fixed a unique $(i-2)$-face $\widecheck{\tau} \in S_{i-2}$, the fourth inequality follows by claim~\ref{bound-non-balanced-faces}, and the last inequality follows since $\lambda \le \eta^{2^{k-i+1}}$ for every $1 \le i \le k-1$. This finishes the proof since
	\begin{equation*}
	\norm{\Upsilon} =
	\sum_{i=1}^{k+1}\norm{\Upsilon_i} \le
	\eta\binom{k+2}{2}2^{k+2}\norm{f}.
	\end{equation*}
\end{proof}

\begin{proof}[Proof of lemma~\ref{small-contribution-of-fat-faces}]
	Let $\sigma \in \overline{S_i}$ and consider a $(k+1)$-face $\tau \in \delta(f\down\sigma)$. By definition there exists a sequence of containments $\tau \supset \tau_k \supset \tau_{k-1} \supset \dotsb \supset \tau_{i+1} \supset \sigma$ such that $\tau_k \in f$ and $\tau_j \in \overline{S_j}$ for every $i < j \le k$. Let us denote $\tau = \{v_0,\dotsc,v_{k+1}\}$ such that $\tau_k = \tau \setminus \{v_{k+1}\}$, and for every $i \le j < k$, $\tau_j = \tau_{j+1}\setminus \{v_{j+1}\}$. Since $f$ is a cocycle, there must exist $\ell \le i$ such that $\tau \setminus \{v_\ell\} \in f$ (otherwise $\tau \in \delta(f)$ in contradiction). Thus, either for every $j \in \{k,k-1,\dotsc,i\}$ it holds that $\tau_j \setminus \{v_\ell\} \in \overline{S_{j-1}}$ or that $\tau \in \Upsilon$ since in this case it contains two $j$-faces from $\overline{S_j}$ such that their intersection belongs to $S_{j-1}$. Therefore,
	\begin{align*}
	\sum_{\sigma \in \overline{S_i}}\norm{(\delta(f\down\sigma), \sigma)} &=
	\sum_{\sigma \in \overline{S_i}}\Pr[\sigma_{k+1} \in \delta(f\down\sigma) \wedge \sigma_i = \sigma] \\[5pt]&=
	\sum_{\sigma \in \overline{S_i}}\Pr[\left(\tau \in \Upsilon \vee (\exists \sigma' \subset \sigma \mbox{ s.t. } \sigma' \in \overline{S_{i-1}} \wedge \sigma_{k+1} \in \Gamma(f\down\sigma'))\right) \wedge \sigma_i = \sigma] \\[5pt]&\le
	\norm{\Upsilon} + (k+1-i)(i+1)\sum_{\sigma' \in \overline{S_{i-1}}}\norm{(f\!\downarrow\!\sigma', \sigma')},
	\end{align*}
	where the $(k+1-i)(i+1)$ factor is due to the probability of choosing $\tau_k$ and $\sigma' = \sigma \setminus \{v_\ell\}$ given that $\tau$ and $\sigma$ were chosen.
	
	Now, since $f$ is locally minimal, $f\down\sigma \subseteq f$ is also locally minimal, hence the $\beta$-coboundary expansion of $X_\sigma$ guarantees that $\norm{(f\down\sigma,\sigma)} \le \beta^{-1}\norm{(\delta(f\down\sigma),\sigma)}$, which completes the proof.
\end{proof}

\subsection{Cosystolic expansion}
We use a similar reduction as in~\cite{KKL14} in order to show that $\delta_1$-expansion of small sets implies cosystolic expansion over any abelian group. Recall that a complex is a cosystolic expander if the following two properties hold: (1) The systems of equations are expanding, i.e., any assignment that does not satisfy all the equations has a large fraction of unsatisfied equations (proportional to the distance from a satisfying assignment). (2) Every cocycle which is not a coboundary is large.

\begin{lemma}[The systems of equations are expanding]\label{lem:cosystolic-abelian-1}
	For any $d,q \in \mathbb{N}$, an abelian group $G$, and $0<\beta<1$ there exist $0 < \lambda, \eta < 1$ such that the following holds: Let $X$ be a $d$-dimensional $q$-bounded degree $\lambda$-local spectral expander with $\beta$-coboundary expanding links over $G$. For any $f \in C^k(X;G)\setminus Z^k(X;G)$, $1 \le k \le d-2$, it holds that
	$$\norm{\delta(f)} \ge \min\left\{\eta^{2^{k+2}-1}, \frac{1}{q\binom{d}{k+1}}\right\}\cdot\dist(f, Z^k(X;G)).$$
\end{lemma}

\begin{proof}
	Let $\varepsilon = 1/2(d+1)^2$, and $\lambda,\eta$ as in theorem~\ref{balancing-algorithm}. If $\norm{\delta(f)} \ge \eta^{2^{k+2}-1}$ we are done. Otherwise, by theorem~\ref{balancing-algorithm}, there exists $f' \in C^k(X;G)$ such that $\dist(f,f') \le q\binom{d}{k+1}\norm{\delta(f)}$, $\norm{\delta(f')} \le \norm{\delta(f)}$, and $\delta(f')$ is $(\eta,\varepsilon)$-non-local. Thus, by corollary~\ref{balanced-cocycles-vanish}, $\delta(f') = 0$, i.e., $f' \in Z^k(X;G)$. Therefore, $\dist(f, Z^k(X;G)) \le \dist(f, f') \le q\binom{d}{k+1}\norm{\delta(f)}$, which completes the proof.
\end{proof}

\begin{lemma}[Every cocycle which is not a coboundary is large]\label{lem:cosystolic-abelian-2}
	For any $d \in \mathbb{N}$, an abelian group $G$, and $0 < \beta < 1$, there exists $0 < \lambda,\eta < 1$ such that the following holds: Let $X$ be a $d$-dimensional $\lambda$-local spectral expander with $\beta$-coboundary expanding links over $G$. For any $f \in Z^k(X;G) \setminus B^k(X;G)$, $0 \le k \le d-1$, it holds that $\norm{f} \ge \eta^{2^d-1}$.
\end{lemma}

\begin{proof}
	Let $\varepsilon = 1/2(d+1)^2$, and $\lambda,\eta$ as promised by proposition~\ref{small-locally-minimal-cocycles-are-balanced}. Assume towards contradiction that there exists $f \in Z^k(X;G)\setminus B^k(X;G)$ with $\norm{f} \le \eta^{2^d-1}$. If $f$ is not minimal, then there exists a minimal $f' \in Z^k(X;G)\setminus B^k(X;G)$ with $\norm{f'} < \norm{f} \le \eta^{2^d-1}$. Since $f' \notin B^k(X;G)$ then $f' \ne \mathbf{0}$. Since $f'$ is locally minimal, by proposition~\ref{small-locally-minimal-cocycles-are-balanced} $f'$ is $(\eta,\varepsilon)$-non-local and hence by corollary~\ref{balanced-cocycles-vanish} $f' = \mathbf{0}$, in contradiction.
\end{proof}

Theorem~\ref{thm:cosystolic-expansion}, which we restate here in a formal way, follows immediately from the above two lemmas.
\begin{theorem}[Cosystolic expansion over any abelian group]
	For any $d,q \in \mathbb{N}$, an abelian group $G$, and $0 < \beta < 1$ there exist $0 < \lambda,\eta <1$ such that the following holds: Let $X$ be a $d$-dimensional $q$-bounded degree $\lambda$-local spectral expander with $\beta$-coboundary expanding links over $G$. Then the $(d-1)$-skeleton of $X$ is an $(\varepsilon, \mu)$-cosystolic expander over $G$, where
	$$\varepsilon = \min\left\{\eta^{2^d-1}, \frac{1}{qd^{d/2}}\right\} \quad\quad\mbox{and}\quad\quad
	\mu = \eta^{2^d-1}.$$
\end{theorem}
\begin{proof}
	Immediate from lemmas~\ref{lem:cosystolic-abelian-1} and~\ref{lem:cosystolic-abelian-2}.
\end{proof}

\section{Result for non-abelian groups}
\subsection{Non-abelian groups}
When the group is non-abelian, the coboundary operator is defined only in dimensions $0$ and $1$, and its definition is more delicate. Let $G$ be a group with a multiplicative operation. The coboundary of a $0$-cochain $f \in C^0(X;G)$ is a $1$-cochain $\delta(f)$ defined by
$$\delta(f)(u,v) = f(u)f(v)^{-1}.$$
The coboundary of a $1$-cochain $g \in C^1(X;G)$ is a $2$-cochain $\delta(g)$  defined by
$$\delta(g)(u,v,w) = g(u,v)g(v,w)g(w,u).$$
One can check that for $f \in C^i(X;G)$, $i \in \{0,1\}$, $\delta(f)$ is an antisymmetric function, i.e., $\delta(f)$ is an $(i+1)$-cochain.

The distance between two cochains $f, g \in C^i(X;G)$ is defined by $\dist(f,g) = \norm{gf^{-1}}$, where $gf^{-1}(\sigma) = g(\sigma)f(\sigma)^{-1}$ for every $\sigma \in \vec{X}(i)$.

Similar to the abelian case, we say that $f \in C^i(X;G)$ is a cocycle if $\delta(f) = \mathbf{1}$.\footnote{It is common to denote the identity element of a multiplicative group by $1$ and not by $0$ as in an additive group.} The distance of a cochain $f \in C^i(X;G)$ from the $i$-cocycles is defined by $$\dist(f, Z^i(X;G)) = \min\{\dist(f,g) \;|\; g \in Z^i(X;G) \}.$$

In order to measure the distance of a $1$-cochain from the $1$-coboundaries, an action of $C^0(X;G)$ on $C^1(X;G)$ is defined, where for $f \in C^0(X;G)$ and $g \in C^1(X;G)$, the definition of $f.g$ is
$$f.g(u,v) = f(u)g(u,v)f(v)^{-1}.$$
Then, the distance of $g$ from the $1$-coboundaries is defined by
$$\dist(g,B^1(X;G)) = \min\{\dist(g,f.g) \;|\; f \in C^0(X;G) \}.$$

\subsection{Weakly-non-local sets and cosystolic expansion}
In the case of a non-abelian group, we cannot get the same non-local property as in abelian groups, rather we get a slightly weaker notion which we call \emph{weakly-non-local}. Roughly speaking, a set of $k$-faces in a given complex is weakly-non-local if its $k$-faces are evenly distributed on their $(k-2)$-subfaces.

\begin{definition}\label{def-3-formal}[Weakly-non-local sets]
	Let $X$ be a $d$-dimensional simplicial complex and $0 < \eta,\varepsilon,\alpha<1$. For any set of $k$-faces $A \subseteq X(k)$, $1 \le k \le d-1$, we define the following set of $(k-2)$-faces:
	$$S_{k-2} = \{\sigma \in X(k-2) \;|\; \norm{A_\sigma} \le \eta \}.$$
	We say that $A$ is $(\eta,\varepsilon,\alpha)$-weakly-non-local if $\norm{S_{k-2}} \ge 1-\varepsilon\norm{A}$ and for every $\tau \in X(k-1)$ it holds that $\norm{A_\tau} \le 1-\alpha$.
\end{definition}

We show that this weakly-non-local property also implies that the set is $\delta_1$-expanding.

\begin{theorem}[Weakly-non-local sets are $\delta_1$-expanding]\label{weakly-balanced-implies-large-delta1}
	Let $X$ be a $d$-dimensional $\lambda$-local spectral expander and $0 < \eta,\varepsilon,\alpha<1$. There exists a constant $c = c(d,\lambda,\eta,\varepsilon,\alpha)$ such that for any $A \subseteq X(k)$, $1 \le k \le d-1$, if $A$ is $(\eta,\varepsilon,\alpha)$-weakly-non-local then
	$$\norm{\delta_1(A)} \ge c\norm{A}.$$
	In particular, if $\varepsilon \le \alpha/3d^3$, $\lambda \le \varepsilon^2$ and $\eta \le \varepsilon^3$ then $\norm{\delta_1(A)} \ge \alpha\norm{A}$.
\end{theorem}

In order to prove theorem~\ref{weakly-balanced-implies-large-delta1} we need a few more definitions and lemmas.
Recall that $S_{k-2}$ is the set of $(k-2)$-faces $\sigma$ satisfying $\norm{A_\sigma} \le \eta$. We also define the following set of $(k-1)$-faces:
$$S_{k-1} = \{\sigma \in X(k-1) \;|\; \norm{A_\sigma} \le \eta^{1/3} \}.$$

We first show that $\norm{\delta_1(A)}$ can be bounded in terms of $\norm{(A,S_{k-1})}$.

\begin{lemma}\label{delta1-decomposition-to-(k-1)-faces}
	Let $X$ be a $d$-dimensional $\lambda$-local spectral expander. For any $A \subseteq X(k)$, $0 \le k \le d-1$, if for every $\sigma \in X(k-1)$ it holds that $\norm{A_\sigma} \le 1-\alpha$ then
	$$\norm{\delta_1(A)} \ge
	(k+1)(k+2)\left((1-\lambda)(1-\alpha-\eta^{1/3})\norm{(A,S_{k-1})} - \Big(\frac{k}{k+1}-(1-\lambda)\alpha\Big)\norm{A} \right).$$
\end{lemma}
\begin{proof}
	Denote by $\delta_i(A)$ the set of $(k+1)$-faces that contain exactly $i$ faces from $A$. By laws of probability
	\begin{equation}\label{delta1-decomposition-to-(k-1)-faces-eq1}
	\begin{aligned}
	(k+2)\norm{f} &=
	(k+2)\sum_{i=0}^{k+2}\Pr[\sigma_{k+1} \in \delta_i(A) \wedge \sigma_k \in A] \\[5pt]&=
	(k+2)\sum_{i=0}^{k+2}\Pr[\sigma_k \in A \;|\; \sigma_{k+1} \in \delta_i(A)]\cdot\Pr[\sigma_{k+1} \in \delta_i(A)] \\[5pt]&=
	(k+2)\sum_{i=0}^{k+2}\frac{i}{k+2}\norm{\delta_i(A)} =
	\sum_{i=1}^{k+2}i\cdot\norm{\delta_i(A)}.
	\end{aligned}
	\end{equation}
	
	Note that any $\tau \in \delta_i(A)$ contains exactly $i$ $k$-faces which are in $A$ and $k+2-i$ $k$-faces which are not in $A$. Therefore,
	\begin{equation}\label{delta1-decomposition-to-(k-1)-faces-eq2}
	i(k+2-i)\norm{\tau} =
	\sum_{\substack{\sigma,\sigma' \subset \tau\\\sigma \in A,\, \sigma' \notin A}}\norm{\tau} =
	\binom{k+2}{k}\sum_{\substack{\sigma,\sigma' \subset \tau\\\sigma \in A,\, \sigma' \notin A}}\norm{(\tau,\sigma\cap\sigma')}.
	\end{equation}
	
	For any pair $\sigma, \sigma' \subset \tau$ such that $\sigma \in A$ and $\sigma' \notin A$, their intersection is a $(k-1)$-face such that $\tau$ is seen in the link of $\sigma\cap\sigma'$ as an edge between a vertex in $A$ and a vertex not in $A$. Thus, summing~\eqref{delta1-decomposition-to-(k-1)-faces-eq2} over all $1 \le i \le k+2$ and $\tau \in \delta_i(A)$ yields
	\begin{equation}\label{delta1-decomposition-to-(k-1)-faces-eq3}
	\sum_{i=1}^{k+2}i(k+2-i)\norm{\delta_i(A)} =
	\binom{k+2}{k}\sum_{\sigma \in X(k-1)}\norm{(E(A_\sigma,\overline{A_\sigma}), \sigma)}.
	\end{equation}
	
	Since $X$ is a $\lambda$-local spectral expander, it follows that
	\begin{equation}\label{delta1-decomposition-to-(k-1)-faces-eq4}
	\begin{aligned}
	\sum_{\sigma \in X(k-1)}\norm{(E(A_\sigma,\overline{A_\sigma}), \sigma)} &=
	\sum_{\sigma \in S_{k-1}}\norm{(E(A_\sigma,\overline{A_\sigma}), \sigma)} + \sum_{\sigma \in \overline{S_{k-1}}}\norm{(E(A_\sigma,\overline{A_\sigma}), \sigma)}\\[4pt]&\ge
	2(1-\lambda)(1-\eta^{1/3})\norm{(A,S_{k-1})} + 2(1-\lambda)\alpha\norm{(A,\overline{S_{k-1}})}\\[12pt]&=
	2(1-\lambda)(1-\alpha-\eta^{1/3})\norm{(A,S_{k-1})} + 2(1-\lambda)\alpha\norm{A}.
	\end{aligned}
	\end{equation}
	
	Combining all together yields
	\begin{align*}
	\norm{\delta_1(A)} &\ge
	\sum_{i=1}^{k+2}i(2 - i)\norm{\delta_i(A)} \\[2pt]&=
	\sum_{i=1}^{k+2}i(k+2-i)\norm{\delta_i(A)} - k\sum_{i=1}^{k+2}i\cdot\norm{\delta_i(A)} \\[5pt]&=
	\binom{k+2}{k}\sum_{\sigma \in X(k-1)}\norm{(E(A_\sigma,\overline{A_\sigma}), \sigma)} - k(k+2)\norm{A} \\[2pt]&\ge
	\binom{k+2}{k}\Big(2(1-\lambda)(1-\alpha-\eta^{1/3})\norm{(A,S_{k-1})} + 2(1-\lambda)\alpha\norm{A}\Big) - k(k+2)\norm{A} \\[6pt]&=
	(k+1)(k+2)\left((1-\lambda)(1-\alpha-\eta^{1/3})\norm{(A,S_{k-1})} - \Big( \frac{k}{k+1}-(1-\lambda)\alpha\Big)\norm{A} \right),
	\end{align*}
	where the second equality follows by~\eqref{delta1-decomposition-to-(k-1)-faces-eq1} and~\eqref{delta1-decomposition-to-(k-1)-faces-eq3}, and the second inequality follows by ~\eqref{delta1-decomposition-to-(k-1)-faces-eq4}.
\end{proof}

Our aim now is to show that $\norm{(A,S_{k-1})} \approx \norm{A}$.
Denote by $\Upsilon$ the set of $k$-faces $\sigma \in A$ that contain two $(k-1)$-faces $\tau,\tau' \in \overline{S_{k-1}}$ such that $\tau \cap \tau' \in S_{k-2}$. We show next that $\norm{\Upsilon}$ is a negligible fraction of $\norm{A}$.

\begin{lemma}\label{small-degenerate-faces}
	Let $X$ be a $d$-dimensional $\lambda$-local spectral expander. For any $A \subseteq X(k)$, $1 \le k \le d-1$, it holds that
	\begin{equation*}
	\norm{\Upsilon} \le \binom{k+1}{2}\left(\eta^{1/3}+\lambda\eta^{-1/3}\right)\norm{A}.
	\end{equation*}
\end{lemma}
\begin{proof}
	For each $\sigma \in \Upsilon$ there exists at least one $(k-2)$-face $\tau \subset \sigma$, $\tau \in S_{k-2}$, such that in the link of $\tau$, $\sigma$ is an edge between two vertices from $\overline{S_{k-1}}$. Fix for each $\sigma$ such a face $\tau$ and denote by $\Upsilon(\tau)$ the $k$-faces in $\Upsilon$ which fixed $\tau$. It follows that
	\begin{align*}
	\norm{\Upsilon} &=
	\sum_{\tau \in S_{k-2}}\norm{\Upsilon(\tau)} =
	\binom{k+1}{k-1}\sum_{\tau \in S_{k-2}}\norm{(\Upsilon(\tau),\tau)} \\&\le
	\binom{k+1}{k-1}\sum_{\tau \in S_{k-2}}\norm{(E((\overline{S_{k-1}})_\tau),\tau)} \\&\le
	\binom{k+1}{k-1}\sum_{\tau \in S_{k-2}}(\norm{(\overline{S_{k-1}})_\tau} + \lambda)\norm{((\overline{S_{k-1}})_\tau, \tau)} \\&\le
	\binom{k+1}{k-1}\left(\eta^{2/3} + \lambda\right)\norm{\overline{S_{k-1}}}
	\le \binom{k+1}{k-1}\left(\eta^{2/3}+\lambda\right)\eta^{-1/3}\norm{A},
	\end{align*}
	where the second inequality follows by lemma~\ref{many-outgoing-edges}, the third inequality follows since for any $\tau \in S_{k-2}$ it holds that
	$$\eta^{1/3}\norm{(\overline{S_{k-1}})_\tau} \le
	\norm{(A,\overline{S_{k-1}})_\tau} \le
	\norm{A_\tau} \le
	\eta,$$
	and the fourth inequality follows since
	$$\eta^{1/3}\norm{\overline{S_{k-1}}} \le \norm{(A,\overline{S_{k-1}})} \le \norm{A}.$$
\end{proof}

We can now prove theorem~\ref{weakly-balanced-implies-large-delta1}.
\begin{proof}[Proof of theorem~\ref{weakly-balanced-implies-large-delta1}]
	Recall that every $k$-face contains $k+1$ faces of dimension $k-1$. Denote by $A' \subseteq A$ the set of $k$-faces of $A$ for which at least $k$ out of their $(k-1)$-faces belong to $S_{k-1}$. Note that
	\begin{equation*}
	\norm{(A',S_{k-1})} =
	\Pr[\sigma_{k-1} \in S_{k-1} \;|\; \sigma_k \in A']\cdot\Pr[\sigma_k \in A'] \ge
	\frac{k}{k+1}\norm{A'}.
	\end{equation*}
	
	Denote by $\Gamma \subseteq X(k)$ the set of $k$-faces that contain at least one $(k-2)$-face from $\overline{S_{k-2}}$. Note that $A \setminus (\Gamma \cup \Upsilon) \subseteq A'$. Thus,
	\begin{equation*}
	\norm{A'} \ge \norm{A} - \norm{\Gamma} - \norm{\Upsilon} \ge
	\left(1- \binom{k+1}{2}(\varepsilon+\eta^{1/3} + \lambda\eta^{-1/3})\right)\norm{A},
	\end{equation*}
	where the second inequality follows since $A$ is $(\eta,\varepsilon,\alpha)$-weakly-non-local and by lemma~\ref{small-degenerate-faces}. It follows that
	\begin{equation}\label{weakly-balanced-implies-large-delta1-eq1}
	\norm{(A,S_{k-1})} \ge
	\norm{(A',S_{k-1})} \ge
	\frac{k}{k+1}\left(1- \binom{k+1}{2}(\varepsilon+\eta^{1/3} + \lambda\eta^{-1/3})\right)\norm{A}.
	\end{equation}
	
	
	Substituting~\eqref{weakly-balanced-implies-large-delta1-eq1} in lemma~\ref{delta1-decomposition-to-(k-1)-faces} completes the proof.
\end{proof}

An immediate corollary of theorem~\ref{weakly-balanced-implies-large-delta1} is that any weakly-non-local set for which its $\delta_1$ is zero must be empty.
\begin{corollary}\label{balanced-coboundaries-vanish}
	For any $d \in \mathbb{N}$ and $0 < \lambda,\eta,\varepsilon,\alpha < 1$ such that $\varepsilon \le \alpha/3d^3$, $\lambda \le \varepsilon^2$ and $\eta \le \varepsilon^3$ the following holds. Let $X$ be a $d$-dimensional $\lambda$-local spectral expander. For any $A \subseteq X(k)$, $1 \le k \le d-1$, if $A$ is $(\eta,\varepsilon,\alpha)$-weakly-non-local and $\norm{\delta_1(A)} = 0$ then $A=\emptyset$.
\end{corollary}
\begin{proof}
	Since $A$ is $(\eta,\varepsilon,\alpha)$-weakly-non-local, by theorem~\ref{weakly-balanced-implies-large-delta1} it holds that
	\begin{equation*}
	\norm{\delta_1(A)} \ge \alpha\norm{A}.
	\end{equation*}
	
	On the other hand, $\norm{\delta_1(A)} = 0$. It follows that $\norm{A} = 0$, i.e., $A = \emptyset$ as required.
\end{proof}

\subsection{The correction procedure}
From now on we focus on $3$-dimensional complexes and we want to show cosystolic expansion for $1$-cochains. The way we do it is by showing a correction procedure for small $2$-coboundaries. The algorithm gets a $1$-cochain $f$ such that $\norm{\delta(f)}$ is small and returns a cochain $f'$ by making a few changes to $f$ such that $\delta(f')$ is weakly-non-local.

We start by showing that any small and locally minimal $2$-coboundary is weakly-non-local. Then we will show an algorithm that is making a few local changes to a cochain in order to make its coboundary locally minimal and hence weakly-non-local.

\subsubsection{Small and locally minimal coboundaries are weakly-non-local}
We begin by proving the following proposition.
\begin{proposition}[Small and locally minimal $2$-coboundaries are weakly-non-local]\label{non-abelian-thin-in-dimension-0}
	For any $0<\beta,\eta,\varepsilon<1$ and a group $G$ there exists $\lambda \le O(\beta^2\eta^2\varepsilon)$ such that the following holds. Let $X$ be a $3$-dimensional $\lambda$-local spectral expander with $\beta$-coboundary expanding links over $G$. For any $f \in B^2(X;G)$, if $\norm{f} \le \beta\eta/2$ and locally minimal then $f$ is $(\eta,\varepsilon,1/|G|)$-weakly-non-local.
	
	
\end{proposition}

The key point in the proof of proposition~\ref{non-abelian-thin-in-dimension-0} is the observation that the \emph{restriction} of a cochain to the links of a local spectral expander resembles very well the global picture, as we explain next.

For any cochain $f \in C^k(X;G)$, $0 \le k \le d-1$ and a vertex $v \in X(0)$, we define the \emph{restriction} of $f$ to the link of $v$ by $f^v(\sigma) = f(\sigma)$ for every $\sigma \in X_v(k)$, i.e., $f^v$ gives values to the $k$-faces that together with $v$ form a $(k+1)$-face in $X$. Note the difference between the restriction $f^v$ and the localization $f_v$: $f^v$ as defined now is a $k$-cochain in the link of $v$, and $f_v$ is a $(k-1)$-cochain in the link of $v$.

Recall that for a $2$-cochain $f \in C^2(X;G)$, we defined $S_0$ as the set of vertices $v$ for which $\norm{f_v} \le \eta$. We define a similar set, $S_0'$ as the set of vertices for  which $\norm{f^v} \le \eta$. We show that in a local spectral expander, $\norm{S_0'} \ge 1-O(\norm{f})$.\footnote{For our purpose it is enough to show that almost all of the vertices are sparse (where $S_0'$ represents the set of sparse vertices). A similar argument could show that for cochains of weight bounded from both sides, almost all of the vertices are not too sparse and not too dense, i.e., almost all of the local views resemble very well the global picture.}

\begin{lemma}\label{in-front-seen-with-right-fraction}
	Let $X$ be a $3$-dimensional $\lambda$-local spectral expander and $0 < \eta < 1$. For any $A \subseteq X(2)$, if $\norm{A} \le \eta/2$ then
	$$\norm{S_0'} \ge 1 - O\left(\frac{\lambda}{\eta^2}\right)\norm{A}.$$
\end{lemma}
\begin{proof}
	Define the following graph $G = (V,E)$, where $V = X(2)$, i.e., all triangles of $X$, and 
	$E = \big\{\{t_1,t_2\} \;|\; \exists u \in X(0) \mbox{ s.t. } t_1 \cupdot u, t_2\cupdot u \in X(3) \big\}$, i.e., there is an edge between $t_1$ and $t_2$ if and only if there exists some vertex in $X$ that completes both $t_1$ and $t_2$ to a tetrahedron.
	
	We define a probability distribution on $G$ that corresponds to the probability distribution of $X$ as follows:
	\begin{itemize}
		\item The probability of a vertex $t \in V$ equals the probability of the corresponding triangle $t \in X(2)$.
		\item The probability of an edge $\{t_1,t_2\} \in E$ equals
		$\E_{u \in X(0)}\Pr[t_1 \cupdot u \;|\; u]\cdot\Pr[t_2 \cupdot u \;|\; u]$,
		where all the probabilities are according to the complex $X$.
	\end{itemize}
	
	Since $X$ is a $\lambda$-local spectral expander, by~\cite[Theorem 5.2]{AJT19} $G$ is an $O(\lambda)$-spectral expander, because its adjacency operator is a two steps walk of the swap walk $S_{1,3}$ of~\cite{AJT19}.
	
	Now, define $\mu : X(0) \to \R$ by $\mu(u) = \norm{A^u} = \Pr[t \in A \;|\; t\cupdot u \in X(3)]$. The following holds by laws of probability:
	\begin{equation}\label{in-front-seen-with-right-fraction-eq1}
	\E_{u \in X(0)}[\mu(u)] = \E_{u \in X(0)}\Pr[t \in A \;|\; t\cupdot u \in X(3)] = \Pr[t\in A] = \norm{A}.
	\end{equation}
	
	\begin{equation}\label{in-front-seen-with-right-fraction-eq2}
	\begin{aligned}
	\E_{u \in X(0)}[\mu(u)^2] &= \E_{u \in X(0)}\Pr[t_1 \in A \;|\; t_1\cupdot u \in X(3)]\cdot\Pr[t_2 \in A \;|\; t_2\cupdot u \in X(3)] \\[5pt]&=
	\Pr_{\{t_1,t_2\} \in E}[t_1 \in A \wedge t_2 \in A] =
	\norm{E(A)},
	\end{aligned}
	\end{equation}
	where $E(A)$ is the set of edges $\{t_1,t_2\}$ in $G$ such that both $t_1$ and $t_2$ are in $A$. Since $G$ is an $O(\lambda)$-spectral expander, it follows that
	$\norm{E(A)} \le \norm{A}^2 + O(\lambda)\norm{A}$. Substituting in~\eqref{in-front-seen-with-right-fraction-eq2} and combining~\eqref{in-front-seen-with-right-fraction-eq1} yields
	\begin{equation*}
	\Var_{u\in X(0)}[\mu(u)] = \E_{u \in X(0)}[\mu(u)^2] - \E_{u \in X(0)}[\mu(u)]^2 \le O(\lambda)\norm{A}.
	\end{equation*}
	
	Now, by Chebyshev's inequality
	\begin{align*}
	\norm{S_0'} &=
	\Pr\big[\norm{A^u} \le \eta\big] =
	1 - \Pr\big[\norm{A^u} > \eta\big] \\[8pt]&\ge
	1- \Pr\left[\mu(u) - \frac{\eta}{2} > \frac{\eta}{2}\right] \\[5pt]&\ge
	1- \frac{Var(\mu)}{(\eta/2)^2} \ge
	1 - O\left(\frac{\lambda}{\eta^2}\right)\norm{A},
	\end{align*}
	where the second inequality follows since $\E[\mu(u)] = \norm{A} \le \eta/2$. This completes the proof.
\end{proof}

We now show that in a complex with coboundary expanding links, the localization and the restriction of a cochain to the links are related. 

\begin{lemma}\label{bound-near-by-in-front}
	Let $X$ be a $3$-dimensional simplicial complex such that its links are $\beta$-coboundary expanders over a group $G$. For any $f = \delta(g)$ where $g \in C^1(X,G)$, if $f$ is locally minimal then for every vertex $v \in X(0)$ it holds that
	$$\norm{f_v} \le \beta^{-1}\norm{f^v}.$$
\end{lemma}
\begin{proof}
	Assume towards contradiction that there exists a vertex $v \in X(0)$ such that $\norm{f_v} > \beta^{-1}\norm{f^v}$. We claim that in this case $f_v$ is not minimal in $X_v$.
	
	Since $X_v$ is a $\beta$-coboundary expander then $\dist(g^v, B^1(X_v;G)) \le \beta^{-1}\norm{\delta(g^v)} = \beta^{-1}\norm{f^v}$. Thus, there exists $h \in C^0(X_v;G)$ such that $\norm{h.g^v} \le \beta^{-1}\norm{f^v}$. Define $g' \in C^0(X_v;G)$ by $g'(u) = h(u)g_v(u)^{-1}$. Note that for every edge $(uw) \in X_v(1)$ it holds that
	\begin{equation}
	\begin{aligned}
	g'.f_v(uw) &=
	g'(u)f_v(uw)g'(w)^{-1} \\[5pt]&=
	g'(u)g(vu)g(uw)g(wv)g'(w)^{-1} \\[5pt]&=
	g'(u)g_v(u)g^v(uw)g_v(w)^{-1}g'(w)^{-1} \\[5pt]&=
	h(u)g_v(u)^{-1}g_v(u)g^v(uw)g_v(w)^{-1}g_v(w)h(w)^{-1} \\[5pt]&=
	h(u)g^v(uw)h(w)^{-1} = h.g^v(uw).
	\end{aligned}
	\end{equation}
	
	It follows that
	\begin{equation*}
	\norm{g'.f_v} = \norm{h.g^v} \le \beta^{-1}\norm{f^v} < \norm{f_v},
	\end{equation*}
	in contradiction to the minimality of $f_v$.
\end{proof}

We can now prove proposition~\ref{non-abelian-thin-in-dimension-0}.
\begin{proof}[Proof of proposition~\ref{non-abelian-thin-in-dimension-0}]
	First, since $f$ is locally minimal, it must be that for every edge $e \in X(1)$, $\norm{f_e} \le 1-1/|G|$. Otherwise, $\norm{f_e}$ can be decreased by adding a coboundary to it. Note that there must be a value $g\in G$ that is given to at least $1/|G|$ fraction of the vertices in $X_e$. Add $g^{-1}$ to all vertices in $X_e$ and note that it decreases the weight of $f_e$, in contradiction to the local minimality of $f$.
	
	Second, since $\norm{f} \le \beta\eta/2$, by lemma~\ref{in-front-seen-with-right-fraction} (applied for $\beta\eta$),
	\begin{equation*}
	\Pr\big[\norm{f^v \le \beta\eta}\big] \ge 1- O\left(\frac{\lambda}{(\beta\eta)^2}\right)\norm{f} \ge 1-\varepsilon\norm{f},
	\end{equation*}
	
	Finally, since $f$ is locally minimal, by lemma~\ref{bound-near-by-in-front},
	\begin{equation*}
	\norm{S_0} = \Pr\big[\norm{f_v} \le \eta\big] \ge \Pr\big[\norm{f^v} \le \beta\eta\big] \ge 1-\varepsilon\norm{f}.
	\end{equation*}
	which completes the proof.
\end{proof}

\subsubsection{The correction algorithm}
We show now an algorithm that gets a $1$-cochain $f$ such that $\norm{\delta(f)}$ is small and returns a $1$-cochain $f'$ by making a few changes to $f$ such that $\delta(f')$ is weakly-non-local. The algorithm runs in iterations, where at every iteration it does the following one step of correction.

\begin{lemma}\label{one-step-non-abelian}[One step of correction]
	Let $X$ be a $3$-dimensional simplicial complex and $G$ a group. For any $f \in C^1(X;G)$, if $\delta(f)$ is not locally minimal then there exists a vertex $v \in X(0)$ and $f' \in C^1(X;G)$ such that $\norm{f'f^{-1}} \le 2\norm{v}$ and $\norm{\delta(f')} < \norm{\delta(f)}$.
\end{lemma}
\begin{proof}
	Denote $g = \delta(f)$. Since $g$ is not locally minimal, there exists a vertex $v \in X(0)$ such that $g_v$ is not minimal in $X_v$. By definition there exists $h \in C^0(X_v;G)$ such that $\norm{h.g_v} < \norm{g_v}$. Define $f' \in C^1(X;G)$ by $f'(vu) = h(u)f(vu)$ for every edge $vu \in \vec{X}(1)$ which contains $v$, and $f'(uw) = f(uw)$ for every edge $uw \in \vec{X}(1)$ which does not contain $v$.
	
	Note that $f$ and $f'$ agree on all edges which does not contain $v$, and thus, $\delta(f)$ and $\delta(f')$ agree on all triangles which does not contain $v$. For every triangle that contains $v$ it holds that
	\begin{align*}
	\delta(f')(vuw) &= f'(vu)f'(uw)f'(wv) =
	h(u)f(vu)f(uw)f(wv)h(w)^{-1} = h.g_v(uw).
	\end{align*}
	Therefore, $\norm{\delta(f')} - \norm{\delta(f)} = \norm{h.g_v} - \norm{g_v} < 0$.
	
	Furthermore, 
	\begin{equation*}
	\norm{f'f^{-1}} \le \Pr[\sigma_1 \ni v] = \frac{\Pr[\sigma_1 \ni v \wedge \sigma_0 = v]}{\Pr[\sigma_0 = v \;|\; \sigma_1 \ni v]} \le 2\norm{v},
	\end{equation*}
	which completes the proof.
\end{proof}

We can now prove theorem~\ref{thm:correction-algorithm} which we restate here for general groups.

\begin{theorem}\label{balancing-non-abelian}
	For any $q \in \mathbb{N}$, a group $G$, and  $0<\beta,\varepsilon<1$ such that $\varepsilon \le 1/81|G|$ there exist constants $0 < \lambda,\eta \le \varepsilon$ such that the following holds: Let $X$ be a $3$-dimensional $q$-bounded degree $\lambda$-local spectral expander with $\beta$-coboundary expanding links over $G$. For any $f \in C^1(X;G)$, if $\norm{\delta(f)} \le \beta\eta/2$ then there exists $f' \in C^1(X;G)$ such that $\dist(f, f') \le 2q\norm{\delta(f)}$, $\norm{\delta(f')} \le \norm{\delta(f)}$, and $\delta(f')$ is $(\eta,\varepsilon,1/|G|)$-weakly-non-local.
\end{theorem}
\begin{proof}
	Let $\eta \le \varepsilon^3$ and $\lambda$ as in lemma~\ref{one-step-non-abelian}.
	Apply lemma~\ref{one-step-non-abelian} for $f$ step by step until no more corrections are possible. Since at every step the norm decreases, this process terminates after some $r \ge 0$ steps. Denote by $v_1,v_2,\dotsc,v_r$ the vertices and by $f^{(1)}, f^{(2)}, \dotsc, f^{(r)}$ the $1$-cochains given by applying lemma~\ref{one-step-non-abelian} for $r$ steps, where at step $i$ we apply it for $f^{(i-1)}$.
	
	Let $f' = f^{(r)}$. Since the norm of $\delta(f)$ decreases at every step of correction, it follows that $\norm{\delta(f')} \le \norm{\delta(f)} \le \beta\eta/2$. Furthermore, since no more corrections are possible, it must be that $\delta(f')$ is locally minimal. Thus, by proposition~\ref{non-abelian-thin-in-dimension-0}, $\delta(f')$ is $(\eta,\varepsilon,1/|G|)$-weakly-non-local.
	
	It is left to show that $\norm{f'f^{-1}}$ is proportional to $\norm{\delta(f)}$. By definition, for any $\sigma \in X(2)$ it holds that $\norm{\sigma} \ge \Big(4|X(3)|\Big)^{-1}$, hence $r \le 4|X(3)|\norm{\delta(f)}$. Thus,
	\begin{align*}
	\dist(f,f') &\le
	\dist(f,f^{(1)}) + \dist(f^{(1)}, f^{(2)}) + \dotsb + \dist(f^{(r-1)},f^{(r)}) \\&\le
	\sum_{i=1}^r2\norm{v_i} \le
	8|X(3)|\frac{q}{4|X(3)|}\norm{\delta(f)} \le
	2q\norm{\delta(f)},
	\end{align*}
	which completes the proof.
\end{proof}

\subsection{Cosystolic expansion over any group}
We now show that cosystolic expansion over any group is implied by $\delta_1$-expansion of small sets. Recall that a complex is a cosystolic expander if the following two properties hold: (1) The systems of equations are expanding, i.e., any assignment that does not satisfy all the equations has a large fraction of unsatisfied equations (proportional to the distance from a satisfying assignment). (2) Every cocycle which is not a coboundary is large.
\begin{lemma}[The systems of equations are expanding]\label{lem:cosystolic-non-abelian-1}
	For any group $G$, $q \in \mathbb{N}$ and $0<\beta<1$ there exist $0 < \lambda, \eta < 1$ such that the following holds: Let $X$ be a $3$-dimensional $q$-bounded degree $\lambda$-local spectral expander with $\beta$-coboundary expanding links over $G$. For any $f \in C^1(X;G)\setminus Z^1(X;G)$ it holds that
	$$\norm{\delta(f)} \ge \min\left\{\frac{\beta\eta}{2}, \frac{1}{2q}\right\}\cdot\dist(f, Z^1(X;G)).$$
\end{lemma}
\begin{proof}
	Let $\varepsilon,\lambda,\eta$ be as in theorem~\ref{balancing-non-abelian}. If $\norm{\delta(f)} \ge \beta\eta/2$ we are done. Otherwise, by theorem~\ref{balancing-non-abelian}, there exists $f' \in C^1(X;G)$ such that $\dist(f,f') \le 2q\norm{\delta(f)}$, $\norm{\delta(f')} \le \norm{\delta(f)}$, and $\delta(f')$ is $(\eta,\varepsilon,1/|G|)$-weakly-non-local. Thus, by corollary~\ref{balanced-coboundaries-vanish}, $\delta(f') = 0$, i.e., $f' \in Z^k(X;G)$. Therefore, $\dist(f, Z^k(X;G)) \le \dist(f, f') \le 2q\norm{\delta(f)}$, which completes the proof.
\end{proof}

\begin{lemma}[Every cocycle which is not a coboundary is large]\label{lem:cosystolic-non-abelian-2}
	For any group $G$ and $0 < \beta < 1$, there exist $0 < \lambda,\eta < 1$ such that the following holds: Let $X$ be a $3$-dimensional $\lambda$-local spectral expander with $\beta$-coboundary expanding links over $G$. For any $f \in Z^1(X;G) \setminus B^1(X;G)$ it holds that $\norm{f} \ge \beta\eta/2$.
\end{lemma}
\begin{proof}
	Let $\varepsilon,\lambda,\eta$ as promised by proposition~\ref{non-abelian-thin-in-dimension-0}. Assume towards contradiction that there exists $f \in Z^k(X;G)\setminus B^k(X;G)$ with $\norm{f} \le \beta\eta/2$. If $f$ is not minimal, then there exists a minimal $f' \in Z^k(X;G)\setminus B^k(X;G)$ with $\norm{f'} < \norm{f} \le \beta\eta/2$. Since $f' \notin B^k(X;G)$ then $f' \ne 0$. Now, since $f'$ is locally minimal, by proposition~\ref{non-abelian-thin-in-dimension-0} $f'$ is $(\eta,\varepsilon,1/|G|)$-weakly-non-local and hence by corollary~\ref{balanced-coboundaries-vanish} $f' = 0$, in contradiction.
\end{proof}

Theorem~\ref{thm:cosystolic-expansion}, which we restate here in a formal way for general groups, follows immediately from the above two lemmas.
\begin{theorem}[Cosystolic expansion over any group]\label{thm:cosystolic-expansion-non-abelian}
	For any group $G$, $q \in \mathbb{N}$ and $0 < \beta < 1$ there exist $0 < \lambda,\eta <1$ such that the following holds: Let $X$ be a $3$-dimensional $q$-bounded degree $\lambda$-local spectral expander with $\beta$-coboundary expanding links over $G$. Then the $2$-skeleton of $X$ is an $(\varepsilon, \mu)$-cosystolic expander over $G$, where
	$$\varepsilon = \min\left\{\frac{\beta\eta}{2}, \frac{1}{2q}\right\} \quad\quad\mbox{and}\quad\quad
	\mu = \frac{\beta\eta}{2}.$$
\end{theorem}
\begin{proof}
	Immediate from lemmas~\ref{lem:cosystolic-non-abelian-1} and~\ref{lem:cosystolic-non-abelian-2}.
\end{proof}

\bibliography{Bibliography}

\end{document}